\documentclass[aps,prl,reprint,superscriptaddress,longbibliography, twocolumn, nofootinbib]{revtex4-2}
\usepackage[T1]{fontenc}
\usepackage[utf8]{inputenc}
\usepackage{lmodern}
\usepackage{microtype}
\usepackage{mathtools}
\usepackage{amsmath,amssymb,amsthm,amsfonts}
\usepackage{bm}
\usepackage{bbm}
\usepackage{overpic}
\usepackage[colorlinks = true, 
citecolor = blue,
linkcolor = blue,
urlcolor = blue]{hyperref}
\usepackage{tikz} 
\usetikzlibrary{quantikz2} 
\usepackage[table]{xcolor}
\usepackage{tikz-cd}
\usepackage{ytableau}
\usepackage{graphicx}

\usetikzlibrary{calc,backgrounds}
\usetikzlibrary{arrows.meta}

\DeclareMathOperator{\tr}{Tr}

\newcommand{\ketbra}[2]{{\vert #1 \rangle \langle #2 \vert}}

\newtheorem{theorem}{Theorem}

\allowdisplaybreaks[1]

\begin{document}

\title{Optimal copy complexity of quantum state cloning}

\author{Sangwoo Jeon}
\email{sangw077@gmail.com}
\affiliation{Department of Physics, Korea Advanced Institute of Science and Technology, Daejeon 34141, Korea}

\author{Vaughn Sohn}
\affiliation{Department of Physics, Korea Advanced Institute of Science and Technology, Daejeon 34141, Korea}
\affiliation{Department of Electronic Engineering, Korea Advanced Institute of Science and Technology, Daejeon 34141, Korea}

\author{Changhun Oh}
\email{changhun0218@gmail.com}
\affiliation{Department of Physics, Korea Advanced Institute of Science and Technology, Daejeon 34141, Korea}

\date{\today}

\begin{abstract}
Quantum state cloning is the task of approximately producing additional copies of an unknown quantum state from a finite number of input copies. The optimal cloning fidelity is known exactly for pure states, but no comparable characterization is known for general mixed states. We determine the optimal asymptotic copy complexity for $d$-dimensional states of rank at most $r$: producing $M$ additional copies with worst-case fidelity at least $1-\varepsilon$ requires and is achievable with $N=\Theta(Mrd/\varepsilon)$ input copies. Remarkably, the lower bound already holds for a family of states with a fixed flat spectrum, while the matching upper bound is achieved by random purification followed by optimal pure-state cloning. For $M=1$, we further show that high-fidelity tomography can be coherently converted into cloning with comparable error, revealing an operational origin of the matching cloning and tomography complexities.
\end{abstract}

\maketitle

\textit{Introduction---}
The no-cloning theorem is one of the most fundamental no-go principles in quantum information theory: an unknown quantum state cannot be copied perfectly~\cite{wootters1982single, dieks1982communication}.
Its fundamental role in quantum information theory is reflected in a wide range of settings.
For instance, the security of quantum key distribution, such as the BB84 protocol~\cite{bennett2014quantum}, is intrinsically rooted in this theorem~\cite{woodhead2013quantum}; unlike classical information, an unknown quantum key cannot be copied and stored for independent measurements by an eavesdropper.
Similarly, the ability to clone an arbitrary state would directly violate the uncertainty principle by allowing simultaneous measurements of noncommuting observables.
More recently, unclonability has also become a useful primitive in quantum cryptography, appearing in settings such as quantum money, copy-protection, and related cryptographic tasks~\cite{aaronson2009quantum, aaronson2012quantum, fefferman2025hardness}.

Yet, the impossibility of perfect cloning does not rule out the possibility of \textit{approximate} cloning of quantum states~\cite{buvzek1996quantum}.
For example, although a single unknown qubit state $\ket{\psi}$ cannot be transformed perfectly into $\ket{\psi}\otimes\ket{\psi}$, one can produce an approximate two-copy state with optimal global fidelity $2/3$~\cite{gisin1997optimal}.
The pertinent question is therefore not whether cloning is possible at all, but how well it can be done.
Quantum state cloning formalizes this question as the task of approximately implementing the transformation
\begin{align}
    \rho^{\otimes N} \xrightarrow{\mathrm{approx.}} \rho^{\otimes (N+M)},
\end{align}
given $N$ copies of an unknown state $\rho$.
This problem has been extensively studied for decades~\cite{scarani2005quantum, fan2014quantum}.

Despite this long history, the optimal fidelity of quantum state cloning for general mixed states has, somewhat surprisingly, remained unresolved.
Beyond its foundational significance, this question has recently gained renewed importance in quantum cryptography, where quantitative cloning bounds are closely related to the security of unclonability-based primitives~\cite{fefferman2025hardness}.
Most progress so far has been made in two related but more restricted directions.
The first is the cloning of pure states~\cite{gisin1997optimal, werner1998optimal, bruss1998optimal, cirac1999optimal}, for which the optimal fidelity is known in closed analytic form.
The second is quantum state broadcasting, where the goal is to generate a correlated multipartite state whose individual marginals reproduce the original mixed state~\cite{barnum1996noncommuting, barnum2007generalized, d2005superbroadcasting, dang2007optimal, parzygnat2024virtual}.
Broadcasting, however, is fundamentally different from cloning: it only constrains the marginals, whereas cloning requires the entire output state to approximate the product state $\rho^{\otimes (N+M)}$.

In this work, we resolve this long-standing problem by determining the optimal copy complexity of $N \to N+M$ quantum state cloning for general mixed states in the asymptotic regime.
We consider $d$-dimensional states of rank at most $r$ and ask how many input copies $N$ are necessary and sufficient to produce $M$ additional copies.

More precisely, using Uhlmann fidelity $F(\rho,\sigma)\coloneqq \left(\tr\sqrt{\sqrt{\rho}\sigma\sqrt{\rho}}\right)^2$ as the figure of merit, we ask the following question: 
how many copies $N$ are necessary and sufficient so that there exists a cloning map $\mathcal{T}$ satisfying 
\begin{align}
    \inf_\rho F(\rho^{\otimes (N+M)},\mathcal{T}(\rho^{\otimes N}))\geq1-\varepsilon?
\end{align}
Throughout the work, we derive a lower bound on $N$ using representation-theoretic tools and then a matching upper bound by constructing an asymptotically optimal cloning map, proving the following theorem:
\begin{theorem}\label{thm:bound}
    The copy complexity of cloning $M$ additional copies of $d$-dimensional quantum states of rank at most $r$ with fidelity at least $1-\varepsilon$ is $N=\Theta(Mrd/\varepsilon)$.
\end{theorem}

Theorem~\ref{thm:bound} extends the optimal pure-state cloning scaling to states of arbitrary rank. For $r=1$, it recovers the known pure-state scaling $N=\Theta(Md/\varepsilon)$, while for mixed states the rank contributes an additional multiplicative factor $r$. 
Importantly, the full lower-bound scaling already holds for a family of states with a fixed flat spectrum, showing that the difficulty is already present when only the supporting subspace is unknown.

For $M=1$, this copy complexity reduces to $N = \Theta(rd/\varepsilon)$, which coincides with the optimal sample complexity of rank-$r$ quantum state tomography under fidelity error~\cite{yuen2023improved, scharnhorst2025optimal, pelecanos2025mixed}. 
Motivated by this matching, we also investigate the relationship between optimal cloning and optimal tomography.
In particular, we show that any high-fidelity tomography algorithm can be coherently converted into a cloning map with comparable fidelity error, thereby providing further intuition for why the two tasks have the same asymptotic copy/sample complexity.

\medskip

\textit{Optimal pure-state cloning---}
As a preliminary, we first recall the pure-state cloning problem, following Ref.~\cite{werner1998optimal}.
The task is to approximately transform $N$ copies of an unknown pure state $\rho=\ketbra{\psi}{\psi}$ into $N+M$ copies.
A key simplification in the pure-state setting is that both the input state $\rho^{\otimes N}$ and the target output state $\rho^{\otimes (N+M)}$ are supported on the fully symmetric subspaces.
From the Schur-Weyl perspective~\cite{fulton2013representation}, this means that only one representation sector is relevant: the sector associated with the one-row Young diagram.
Thus, pure-state cloning avoids the main complication that will appear for mixed states, namely, transitions among many different representation sectors~(see Fig.~\ref{fig:map} (a)).
This is what makes pure-state cloning admit a particularly simple and elegant solution.

The optimal $N\to N+M$ pure-state cloning map can be described operationally as follows: append an unnormalized identity operator to the input and project the resulting state onto the symmetric subspace with the appropriate normalization.
More precisely,
\begin{align}
    &\mathcal{T}_\mathrm{Pure}(\rho^{\otimes N})\notag\\
    &\quad\coloneqq\frac{\dim(\mathrm{Sym}_N^d)}{\dim(\mathrm{Sym}_{N+M}^d)}\Pi_{\mathrm{Sym}_{N+M}^d}(\rho^{\otimes N}\otimes {I}_d^{\otimes M})\Pi_{\mathrm{Sym}_{N+M}^d},
\end{align}
where $\mathrm{Sym}_N^d$ denotes the symmetric subspace of $(\mathbb C^d)^{\otimes N}$, $\Pi_{\mathrm{Sym}_N^d}$ denotes the corresponding projector, and $I_d$ is the identity operator on $\mathbb C^d$.
The normalization factor makes the map trace-preserving on the symmetric input subspace.

Here, using the dimension formula for the symmetric subspace, $\dim(\mathrm{Sym}_N^d)=\binom{d+N-1}{N}$, a direct calculation shows that this constant coincides with the cloning fidelity of the optimal pure-state cloning map:
\begin{align}
    F(\rho^{\otimes (N+M)}, \mathcal{T}_{\mathrm{Pure}}(\rho^{\otimes N}))= \prod_{i=1}^M\frac{N+i}{N+d+i-1}
\end{align}
for all pure $\rho$.
Expanding the above expression gives the fidelity scaling $1-\Theta(Md/N)$ in the high-fidelity regime.
Therefore, achieving worst-case fidelity at least $1-\varepsilon$ requires $N=\Theta(Md/\varepsilon)$ copies, establishing the copy complexity of pure-state cloning.

This simple one-sector structure, however, is special to pure states.
For a mixed state $\rho$, the tensor power $\rho^{\otimes N}$ is no longer confined to the symmetric subspace.
Instead, it decomposes over many Schur-Weyl sectors, each labeled by a Young diagram, and a general cloning map can, in principle, transfer weight between these sectors~(see Fig.~\ref{fig:map} (b)).
In fact, the optimal pure-state cloning map above is not even trace-preserving on arbitrary mixed-state inputs.
Thus, mixed-state cloning cannot be obtained by a direct extension of the symmetric-subspace projection argument; one must control all possible sector-to-sector transitions.
Nevertheless, this pure-state cloning map will reappear in our upper bound after we lift the mixed input to random purifications in a larger Hilbert space.

\medskip

\textit{Lower bound---}
Our first main result is that $N=\Omega(Mrd/\varepsilon)$ copies are necessary for cloning $M$ additional copies of all $d$-dimensional states of rank at most $r$. 
The lower bound already follows from a particularly simple family of mixed states: normalized projectors onto $l$-dimensional subspaces, $\rho_l=\Pi_l/l$, with $l=\Theta(r)$. 
These states have a fixed flat spectrum, so the only state-dependent information is the choice of the supporting subspace. 
We show that even for this restricted family, every cloning map must incur an infidelity of order $Mrd/N$ in the high-fidelity regime.

To prove this, we exploit the symmetry of this family. 
By averaging over unitary rotations and permutations, any cloning map can be replaced, without decreasing its worst-case fidelity, by a unitary-covariant and permutation-invariant map. 
Schur-Weyl duality then reduces such maps to elementary transitions between representation sectors, which can be bounded uniformly. 
We provide a proof sketch below and defer the full argument to Supplemental Material (SM) Section S2.

\begin{proof}[Proof sketch of lower bound]
    We prove that every cloning map has worst-case cloning fidelity at most $1-\Omega(Mrd/N)$ over $d$-dimensional states of rank at most $r$.
    By the symmetry reduction below, it suffices to prove this upper bound for unitary-covariant and permutation-invariant cloning maps.

    Given an arbitrary cloning map $\mathcal{T}$, define its unitary symmetrization by
    \begin{align}
        \bar{\mathcal{T}}(\rho^{\otimes N})\coloneqq \int U^{\dagger\otimes (N+M)}\mathcal{T}((U\rho U^\dagger)^{\otimes N})U^{\otimes (N+M)}dU
    \end{align}
    where the integration is taken with respect to the Haar measure.
    Since the ideal target $\rho^{\otimes (N+M)}$ transforms covariantly under the same unitary action and fidelity is concave, this symmetrization cannot decrease the worst-case fidelity, \textit{i.e.},
    \begin{align}
        \inf_\rho F(\rho^{\otimes (N+M)}, \mathcal{T}(\rho^{\otimes N}))\leq \inf_\rho F(\rho^{\otimes (N+M)}, \bar{\mathcal{T}}(\rho^{\otimes N})).
    \end{align}
    Moreover, the symmetrized map is unitary covariant by construction.
    Averaging over input and output permutations in the same way further allows us to assume permutation invariance.

    Imposing these symmetries significantly reduces the degrees of freedom of the cloning map.
    Under Schur-Weyl duality, the input and output Hilbert spaces decompose into sectors labeled by Young diagrams $\lambda\vdash_d N$ and $\mu\vdash_d N+M$, respectively.
    This allows us to write a cloning map as
    \begin{align}
        \mathcal{T}(\rho^{\otimes N})=\sum_\lambda\mathcal{T}_\lambda(\rho^{\otimes N}|_\lambda)
    \end{align}
    where $\mathcal{T}_\lambda$ acts only on the $\lambda$-irrep sector of the input Hilbert space and $\rho^{\otimes N}|_\lambda$ denotes the input state $\rho^{\otimes N}$ restricted to this sector.
    Using the classification of unitary-covariant and permutation-invariant channels~\cite{manvcinska2025classification}, each $\mathcal T_\lambda$ is a convex combination of elementary component maps connecting allowed input and output representation sectors. 
    Hence, it is sufficient to bound the cloning fidelity of each component map separately; the same bound then holds for an arbitrary symmetric cloning map by convexity. 
    The explicit form of these component maps is given in the Supplemental Material.

    The final step is to bound each component map on the hard instance $\rho_l=\Pi_l/l$ introduced above. 
    For this state, the overlap of each component output with the support of the ideal clone can be bounded in terms of the first-row length $\lambda_1$ of the input Young diagram $\lambda$. 
    Averaging these bounds over the Schur-Weyl distribution of $\rho_l^{\otimes N}$ yields
    \begin{align}
    F(\rho_l^{\otimes (N+M)},\mathcal{T}(\rho_l^{\otimes N}))
    \le
    1-\Omega\left(\frac{Mrd}{N}\right)
    \end{align}
    in the high-fidelity regime. 
    Since the bound holds uniformly for every component map, it also holds for any convex combination of them, and hence for every cloning map after the symmetry reduction.

\end{proof}

\begin{figure}[t]
    \centering
    \includegraphics[width=\linewidth]{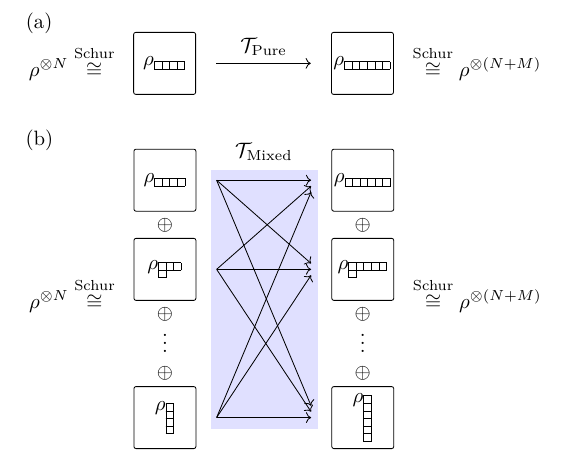}
    \caption{Representation-theoretic structure of unitary-covariant and permutation-invariant $N \to N+M$ cloning maps, illustrated for $M=2$. (a) In pure-state cloning, the input and ideal output states lie in the symmetric subspace, so the optimal cloning map reduces to the fixed map $\mathcal{T}_{\mathrm{Pure}}$. (b) For general mixed-state cloning, the focus of this work, the cloning map $\mathcal{T}_{\mathrm{Mixed}}$ is expressed as a convex combination of component maps between irrep sectors.}
    \label{fig:map}
\end{figure}

For $M=1$, our lower bound reduces to $N=\Omega(rd/\varepsilon)$, matching the optimal sample-complexity lower bound for rank-$r$ quantum state tomography under fidelity error~\cite{yuen2023improved, scharnhorst2025optimal}. 
Interestingly, flat-spectrum states with an unknown supporting subspace, as used in our lower-bound construction, also form hard instances for tomography~\cite{scharnhorst2025optimal}. 
This parallel suggests a close connection between cloning and tomography, which we make more precise later.

\medskip

\textit{Upper bound---}
We now prove the matching upper bound by constructing an explicit cloning map based on random purification~\cite{tang2025conjugate, pelecanos2025mixed}.
The main idea is to lift a rank-$r$ mixed state to a pure state on a larger Hilbert space of dimension $rd$, apply the optimal pure-state cloning map there, and then trace out the purification register.
This yields a cloning map with fidelity at least $\prod_{i=1}^M (N+i)/(N+rd+i-1)$, and hence achieves $N=O(Mrd/\varepsilon)$.
We give the proof sketch here and defer the full proof to SM Section S3.

\begin{proof}[Proof sketch of upper bound]
    The construction is guided by the optimal pure-state cloning map reviewed above.
    If the mixed input $\rho^{\otimes N}$ could be represented as $N$ copies of a pure state, then applying the pure-state cloning map would immediately produce $M$ additional copies with fidelity governed by the symmetric-subspace dimension ratio.
    The obstacle is that $\rho^{\otimes N}$ itself is not supported on a single symmetric subspace.
    We overcome this obstacle by first applying the recently introduced {\it random purification} map~\cite{pelecanos2025mixed}, which lifts $\rho^{\otimes N}$ to a mixture of $N$ identical pure-state purifications in a larger Hilbert space.

    More precisely, let $\ket{\Psi_{\rho}}\in\mathbb{C}^{d}\otimes\mathbb{C}^{r}$ be a purification of a state $\rho$ of rank at most $r$, and define $\ket{\Psi_{\rho,U}}=(I_d\otimes U)\ket{\Psi_{\rho}}$ for $U\in\mathrm{U}(r)$.
    The random purification map is defined by
    \begin{align}
        \Phi(\rho^{\otimes N})\coloneqq\mathbb{E}_U \ketbra{\Psi_{\rho,U}}{\Psi_{\rho,U}}^{\otimes N},
    \end{align}
    which is a completely positive and trace-preserving map, where the expectation is taken over Haar-random unitaries $U$.
    Here, the output state lies entirely in the symmetric subspace of the $N$-copy Hilbert space $(\mathbb C^{rd})^{\otimes N}$.

    Now consider the cloning map obtained by sequentially applying the random purification channel, the optimal pure-state $N \to N+M$ cloning map, and the partial trace over the purification registers~$R$:
    \begin{align}
        \mathcal{T}_{\mathrm{RP}}(\rho^{\otimes N})
        \coloneqq
        \tr_{R}\mathcal{T}_\mathrm{Pure}(\Phi(\rho^{\otimes N})).
    \end{align}
    The cloning fidelity of the random-purification cloning map $\mathcal{T}_{\mathrm{RP}}$ is then lower bounded by that of the optimal pure-state cloning map in dimension $rd$, namely,
    \begin{align}
        F(\rho^{\otimes (N+M)}, \mathcal{T}_{\mathrm{RP}}(\rho^{\otimes N}))\geq \prod_{i=1}^M \frac{N+i}{N+rd+i-1}
    \end{align}
    for all $\rho$.
    This yields the copy complexity $N=O(Mrd/\varepsilon)$, completing the proof.
\end{proof}

Together with the lower bound, this shows that the random-purification cloning map $\mathcal{T}_{\mathrm{RP}}$ is asymptotically optimal. 
The upper bound also gives a simple operational interpretation of the $rd$ dependence: random purification lifts a rank-$r$ state from dimension $d$ to a pure state in dimension $rd$, where the optimal pure-state cloner can be applied, while the matching lower bound shows that the resulting rank dependence is unavoidable.

The constructive nature of the upper bound also allows us to characterize the gate complexity of the resulting cloning protocol. 
The gate complexity of the random purification map is dominated by the Schur transform and is given by $\mathrm{poly}(N,\log rd)$ in the regime $rd\gg N$~\cite{harrow2005applications, pelecanos2025mixed}. 
Meanwhile, the best-known gate complexity for the optimal pure-state cloning map is $O((N+M)rd\operatorname{polylog}(N+M,rd,1/\varepsilon))$~\cite{manvcinska2025classification}. 
Combining these two gives a concrete circuit implementation of the random-purification cloning protocol, in addition to its optimal copy-complexity scaling.

This construction also parallels the optimal tomography algorithm, which similarly applies random purification to the input state $\rho^{\otimes N}$ followed by the optimal pure-state tomography algorithm.

\medskip

\textit{Relation with tomography---}
We now specialize to the case $M=1$.
In this regime, Theorem~\ref{thm:bound} gives the copy complexity $N=\Theta(rd/\varepsilon)$, matching the optimal sample complexity of rank-$r$ quantum state tomography for worst-case expected infidelity~\cite{yuen2023improved,scharnhorst2025optimal,pelecanos2025mixed}.
Our result reveals that this matching reflects more than a coincidence in asymptotic scaling: a high-fidelity tomography procedure can be implemented coherently and converted into a high-fidelity cloning map without increasing the number of input copies.

Relations between cloning and state estimation have long been studied in asymptotic settings~\cite{bae2006asymptotic, chiribella2010quantum, yang2013global, harrow2013church}, while recent works have also highlighted computational and cryptographic distinctions between the two tasks~\cite{nehoran2023computational, bostanci2025general,fefferman2025hardness, bansal2026cloning}.
In the information-theoretic setting considered here, we establish the following direct reduction:
\begin{theorem}
    For any class of quantum states, suppose there exists a tomography algorithm that estimates the state with fidelity at least $1-\varepsilon$ using $N$ samples with $O(\varepsilon)$ failure probability.
    Then there exists an $N\to N+1$ cloning map that achieves fidelity at least $1-c\varepsilon$ for some constant $c$, using the same $N$ input copies.
\end{theorem}
\noindent
The assumption on the failure probability is not restrictive for the asymptotic analysis of the copy complexity, since the optimal tomography algorithm achieves a failure probability that decays exponentially with $N$~\cite{pelecanos2025mixed}. 
In general, the failure probability contributes linearly to the fidelity error of the resulting cloning maps.

The theorem gives a direct operational meaning to the matching copy complexities of cloning and tomography. 
Although tomography is usually viewed as a measurement process that converts quantum states into classical information, it can instead be implemented coherently: the estimate can be used to prepare an additional copy, after which the tomography procedure is uncomputed. 
In this way, information extracted for state estimation can be reused to generate a new quantum copy without increasing the number of input states.

This theorem can be viewed as a mixed-state analogue of the tomography-to-cloning implication of Ref.~\cite{fefferman2025hardness}, which was proved for pure states. For mixed states, the coherent-measurement construction is combined with purifications of the estimated states to control the global cloning fidelity. We provide a proof sketch here and give the full proof in SM Section S4.

\begin{proof}[Proof sketch]
    We construct a map that sequentially performs a coherent tomography measurement, controlled preparation of an additional copy, and an \textit{uncomputation} of the measurement process.
    Let $\{K_x\}$ be Kraus operators for the tomography POVM, and let $\hat{\rho}_x$ be the estimate produced upon observing outcome $x$.    
    The isometry
    \begin{align}
        V = \sum_x K_x \otimes \ket{x}
    \end{align}
    then gives a coherent implementation of this tomography procedure.
    Define the cloning map by the following circuit:
    \begin{figure}[h!]
        \centering
        \includegraphics{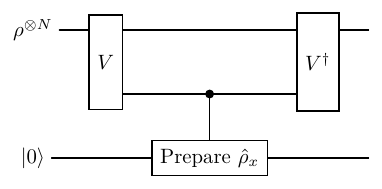}
        \caption{An $N \to N+1$ cloning map constructed from a tomography algorithm.}
    \end{figure}
    
    \noindent
    The circuit applies $V$, coherently prepares an additional system in the estimate $\hat{\rho}_x$, and then applies $V^\dagger$.
    Using purifications of $\rho$ and $\hat{\rho}_x$ together with Uhlmann's theorem, one obtains fidelity at least $1-O(\varepsilon)$ with the ideal cloned state $\rho^{\otimes (N+1)}$.
\end{proof}

The reduction has an immediate consequence: any lower bound on the copy complexity of $M=1$ cloning also applies to the sample complexity of tomography. 
Consequently, our cloning lower bound yields $N=\Omega(rd/\varepsilon)$ for tomography, providing an alternative route to the optimal bound recently established in Ref.~\cite{scharnhorst2025optimal}. 
Conversely, known optimal tomography algorithms with sample complexity $N=O(rd/\varepsilon)$ also yield cloning protocols with the same scaling. 
The random-purification construction developed above provides a direct cloning protocol and, unlike this reduction, applies naturally to arbitrary $M$.

\medskip

\textit{Conclusion---}
In this work, we established the optimal copy complexity $N=\Theta(Mrd/\varepsilon)$ for cloning $M$ additional copies of an arbitrary $d$-dimensional state of rank at most $r$. 
The matching lower and upper bounds show that the additional rank dependence is fundamental: it already appears for states with a fixed flat spectrum, while random purification provides an explicit protocol attaining the optimal scaling. 
For $M=1$, we further uncovered an operational connection between cloning and tomography through a coherent tomography-to-cloning reduction. 
An important open question is to determine the exact optimal cloning fidelity beyond its asymptotic scaling.
Another interesting direction is the optimal cloning of more restricted classes of states, \textit{e.g.}, states with a limited amount of magic or states generated by shallow circuits.

\medskip
\paragraph*{Note added.}
During the preparation of this manuscript, the independent work of Li, Theil, Harrow, and Chuang~\cite{li2026coherent} appeared, which also uses random purification together with pure-state cloning in the context of coherent quantum inference. 
Our work establishes the matching lower bound under worst-case global fidelity, thereby determining the optimal asymptotic copy complexity of mixed-state cloning.

\medskip
\paragraph*{AI use statement.}
The authors used OpenAI ChatGPT (GPT-5.5) to assist in checking the mathematical proofs and revising the manuscript. In this process, the model identified an error in an earlier proof, and the authors used subsequent discussions with the model to help diagnose and revise the argument. The corrected proof and all results in the manuscript were independently verified by the authors, who take full responsibility for the content.

\bigskip

\begin{acknowledgements}
This research was supported by the National Research Foundation of Korea Grants (No. RS-2024-00431768 and No. RS-2025-00515456) funded by the Korean government (Ministry of Science and ICT (MSIT)) and the Institute of Information \& Communications Technology Planning \& Evaluation (IITP) Grants funded by the Korean government (MSIT) (No. RS-2024-00437284, No. IITP-2025-RS-2025-02283189 and No. IITP-2025-RS-2025-02263264). This work was supported by Global Partnership Program of Leading Universities in Quantum Science and Technology (RS-2025-08542968) through the National Research Foundation of Korea~(NRF) funded by the Korean government (Ministry of Science and ICT(MSIT)).
\end{acknowledgements}

\bibliography{reference}

\end{document}


\title{Supplemental Materials: Optimal copy complexity of quantum state cloning}

\author{Sangwoo Jeon}
\email{sangw077@gmail.com}
\affiliation{Department of Physics, Korea Advanced Institute of Science and Technology, Daejeon 34141, Korea}

\author{Vaughn Sohn}
\affiliation{Department of Physics, Korea Advanced Institute of Science and Technology, Daejeon 34141, Korea}
\affiliation{Department of Electronic Engineering, Korea Advanced Institute of Science and Technology, Daejeon 34141, Korea}

\author{Changhun Oh}
\email{changhun0218@gmail.com}
\affiliation{Department of Physics, Korea Advanced Institute of Science and Technology, Daejeon 34141, Korea}

\date{\today}

\maketitle

\tableofcontents

\section{Review of representation theory}

We provide a brief review of the representation-theoretic tools used below, referring the reader to standard references for further details~\cite{brian2003lie, goodman2009symmetry, fulton2013representation, hayashi2017group}.
We begin with a few basic definitions.

\begin{definition}[Representations]\label{def:rep}
    A representation of a group $G$ is a pair $(\varphi,V)$, where $V$ is a vector space and $\varphi:G\to \mathrm{GL}(V)$ is a group homomorphism.
    Here, $\mathrm{GL}(V)$ denotes the group of invertible linear maps on $V$. 
    The dimension of $V$, $\dim(V)$, is called the dimension of the representation.
\end{definition}
\noindent
Throughout, all representations are on finite-dimensional complex vector spaces.
We denote the representation $(\varphi,V)$ by $\varphi$ or by $V$ when it is clear from the context.
\begin{definition}[Irreducible representations]\label{def:irrep}
    For a representation $(\varphi, V)$ of a group $G$, a subspace $W \subseteq V$ is called an invariant subspace of $V$ if $\varphi(g) W \subseteq W$ holds for all $g \in G$.
    If such a nontrivial subspace $W$ (i.e., $W \ne V$ and $W \ne \{0\}$) exists, $(\varphi, V)$ is called a reducible representation. Otherwise, it is called an irreducible representation, or an irrep for short.
\end{definition}
\begin{definition}[Isomorphic representations]
    Two representations $(\varphi_1, V_1)$ and $(\varphi_2, V_2)$ of a group $G$ are called isomorphic (or equivalent), written $\varphi_1\cong\varphi_2$, if there exists an invertible linear map $T : V_1 \to V_2$ such that
    $T \varphi_1(g) = \varphi_2(g) T$ for all $g \in G$.
\end{definition}
\begin{definition}[Intertwiner]\label{def:intertwiner}
    For two representations $(\varphi_1, V_1)$ and $(\varphi_2, V_2)$ of a group $G$, a linear map $T : V_1 \to V_2$ is called an intertwiner if $T \varphi_1(g) = \varphi_2(g) T$ for all $g \in G$.
\end{definition}
\noindent
We denote by $\mathrm{Hom}_G(V_1, V_2)$ the set of all intertwiners from $V_1$ to $V_2$.

We focus on the irreducible representations of the permutation group $\mathrm{S}_N$ and the unitary group $\mathrm{U}(d)$.
These representations are indexed by partitions and staircases, respectively, which we define as follows:

\begin{definition}[Staircases and Partitions]\label{def:partition}
    A staircase $\lambda=(\lambda_1,\dots,\lambda_k)$ is a finite tuple of integers satisfying $\lambda_1\geq\dots\geq \lambda_k$.
    A staircase $\lambda$ is said to be a partition if all its entries are nonnegative.
    If a partition $\lambda$ has size $N$, \textit{i.e.}, $|\lambda|\coleq\sum_i \lambda_i=N$, we write $\lambda\vdash N$.
    The length $\len(\lambda)$ is the number of nonzero entries of a partition $\lambda$.
    If a partition $\lambda$ has size $N$ with $\len(\lambda)\leq l$, we write $\lambda\vdash_l N$.
    We identify partitions that differ only by trailing zeros.
\end{definition}
\noindent
Throughout, unless otherwise specified, the entries of a staircase or partition are indexed by subscripts, as in $\lambda=(\lambda_1,\dots,\lambda_k)$.
Partitions and staircases are often represented by (generalized) Young diagrams, which are arrays of boxes.
\begin{figure}[h]
    \centering
    \includegraphics{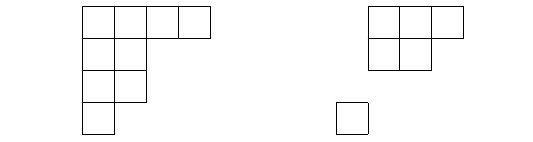}
    \caption{Examples of a Young diagram of a partition $\lambda=(4,2,2,1)$ and a generalized Young diagram of a staircase $\lambda=(3,2,0,-1)$.}
    \label{fig:young}
\end{figure}

\noindent
With this notation, every irrep of $\mathrm{S}_N$ is labeled by a partition $\lambda\vdash N$, and the corresponding representation $\mathrm{Sp}_\lambda$ is called the \textit{Specht module}~\cite[Lecture~4]{fulton2013representation}.
Similarly, every irrep of $\mathrm{U}(d)$ is labeled by a length-$d$ staircase $\lambda$, and the corresponding representation $\mathrm{V}_\lambda^d$ is called the \textit{Weyl module}~\cite[Sec.~4.3.2]{hayashi2017group}.

We refer the reader to the references for rigorous definitions of the Specht and Weyl modules, and focus here on the structure of the Weyl modules needed below.
Restricting an irrep of $\mathrm{U}(d)$ to the subgroup $\mathrm{U}(d-1)\times \mathrm{U}(1)$ yields the following branching rule~\cite[Theorem~8.1.1]{goodman2009symmetry}:
\begin{align}
    \mathrm{V}_{\lambda}^d\downarrow_{\mathrm{U}(d-1)\times \mathrm{U}(1)}
    \cong \bigoplus_{\lambda'\prec\lambda}\mathrm{V}_{\lambda'}^{d-1}\otimes \mathbb{C}_{|\lambda|-|\lambda'|}.
    \label{eq:u_restrict}
\end{align}
Here, $\lambda'\prec \lambda$ denotes the interlacing condition: for the staircases $\lambda$ and $\lambda'$,
\begin{align}
    \lambda_d\leq \lambda'_{d-1}\leq \lambda_{d-1}\leq\dots\leq \lambda'_1\leq\lambda_1
\end{align}
holds.
The representation $\mathbb{C}_{|\lambda|-|\lambda'|}$ denotes the one-dimensional representation of $\mathrm{U}(1)$ defined by $z\mapsto z^{|\lambda|-|\lambda'|}$.
Successively applying the multiplicity-free branching rule along
$\mathrm{U}(d)\to\mathrm{U}(d-1)\to\dots\to\mathrm{U}(1)$
yields an orthogonal decomposition into one-dimensional subspaces indexed by chains of interlacing staircases.
Thus, each such subspace is indexed by a \textit{Gelfand-Tsetlin (GT)} pattern, a triangular array of integers whose adjacent rows satisfy the interlacing conditions, defined as follows:

\begin{definition}[GT pattern~\cite{gelfand1950finite}]
    A GT pattern is a triangular array of integers of the form
    \begin{align}
        \Lambda=
        \left(
        \begin{array}{ccccccc}
        \lambda^d_1 && \lambda^d_2 && \cdots && \lambda^d_d \\
        & \lambda^{d-1}_1 && \cdots && \lambda^{d-1}_{d-1} & \\
        &&& \vdots &&& \\
        &&& \lambda^{1}_{1} &&&
        \end{array}
        \right),
    \end{align}
    where the interlacing condition $\lambda^i_j\geq\lambda^{i-1}_j\geq\lambda^i_{j+1}$ holds for all $i=2,\dots,d$ and $j=1,\dots,i-1$.
\end{definition}
\noindent
The GT patterns with top row $\lambda$ index a basis of the Weyl module $\mathrm{V}^d_\lambda$, called the GT basis.
The action of $\mathrm{U}(d)$ on this basis can be described explicitly through the corresponding Lie algebra representation~\cite{gelfand1950finite, molev2006gelfand}.
For our purposes, however, we do not need the full explicit form of the $\mathrm{U}(d)$ action.
We instead introduce the action of the central elements $zI_d$ for $z\in\mathrm{U}(1)$, which will be used below.
The central element $zI_d$ acts on $\mathrm{V}_{\lambda}^d$ as
\begin{align}
    \varphi_{\lambda}(zI_d)=z^{|\lambda|}I_{\mathrm{V}_\lambda^d}.
    \label{eq:z_indentity_cond}
\end{align}

We now recall the standard representation-theoretic results needed below.
We first recall Schur's lemma, which characterizes intertwiners between irreps:

\begin{lemma}[Schur's lemma]
    Let $(\varphi_1,V_1)$ and $(\varphi_2,V_2)$ be irreps of a group $G$.
    \begin{itemize}
        \item If $\varphi_1\ncong \varphi_2$, then $\mathrm{Hom}_G(V_1,V_2)=\{0\}$.
        \item If $\varphi_1\cong \varphi_2$, then $\mathrm{Hom}_G(V_1,V_2)=\{cT \mid c\in\mathbb{C}\}$ for any fixed nonzero intertwiner $T:V_1\to V_2$.
    \end{itemize}
    In particular, if $(\varphi_1,V_1)=(\varphi_2,V_2)$, then every intertwiner $T:V_1\to V_1$ is of the form $T=cI_{V_1}$ for some $c\in\mathbb{C}$.
\end{lemma}

We next recall Schur-Weyl duality for the $N$-copy qudit Hilbert space $(\mathbb{C}^d)^{\otimes N}$.
The Hilbert space carries two natural commuting group actions: the collective action of the unitary group $\mathrm{U}(d)$ and the permutation action of the symmetric group $\mathrm{S}_N$.
For $U\in \mathrm{U}(d)$, define the collective unitary action $\varphi_{\mathrm{U}(d)}(U)$ on $(\mathbb{C}^d)^{\otimes N}$ by
\begin{equation}
    \varphi_{\mathrm{U}(d)}(U)\ket{i_1\cdots i_N}
    \coleq
    U^{\otimes N}\ket{i_1\cdots i_N},
\end{equation}
which defines a representation of $\mathrm{U}(d)$ on $(\mathbb{C}^d)^{\otimes N}$.
Similarly, for $\pi\in \mathrm{S}_N$, define the permutation action $\varphi_{\mathrm{S}_N}(\pi)$ by
\begin{equation}
    \varphi_{\mathrm{S}_N}(\pi)\ket{i_1\cdots i_N}
    \coleq
    \ket{i_{\pi^{-1}(1)}\cdots i_{\pi^{-1}(N)}},
\end{equation}
which defines a representation of $\mathrm{S}_N$ on $(\mathbb{C}^d)^{\otimes N}$.
Under these commuting actions, the tensor-product Hilbert space admits the following decomposition.

\begin{lemma}[Schur-Weyl duality]
    As a representation of $\mathrm{S}_N\times\mathrm{U}(d)$,
    \begin{align}
        (\mathbb{C}^d)^{\otimes N}
        \overset{\mathrm{Schur}}{\cong}
        \bigoplus_{\lambda\vdash_d N}
        \mathrm{Sp}_\lambda\otimes \mathrm{V}_\lambda^d.
    \end{align}
    The corresponding unitary isomorphism is called the Schur transform.
\end{lemma}

A standard tool for analyzing tensor products of representations is the Clebsch-Gordan~(CG) decomposition, which characterizes their decomposition into irreducible representations as follows:
\begin{lemma}[Clebsch-Gordan~(CG) decomposition]\label{lem:cg}
    For staircases $\lambda$ and $\nu$, the tensor product of the corresponding Weyl modules $\mathrm{V}_\lambda^d$ and $\mathrm{V}_\nu^d$ decomposes as
    \begin{align}
        \mathrm{V}_\lambda^d \otimes \mathrm{V}_\nu^d
        \cong
        \bigoplus_{\mu}
        \mathrm{V}_\mu^d \otimes \mathbb{C}^{c^\mu_{\lambda,\nu}},
    \end{align}
    where the sum is over length-$d$ staircases $\mu$, and $c^\mu_{\lambda,\nu}$ is the \textit{Littlewood-Richardson~(LR) coefficient}.
\end{lemma}
\noindent
A unitary intertwiner implementing the corresponding component of this decomposition is called a CG transformation:
\begin{align}
    U_{\mu\to \lambda,\nu}^{\mathrm{CG}}
    :
    \mathrm{V}_\mu^d\otimes \mathbb{C}^{c_{\lambda,\nu}^{\mu}}
    &\to
    (\mathrm{V}_\lambda^d\otimes \mathrm{V}_\nu^d)|_{\mu},
\end{align}
where $(\mathrm{V}_\lambda^d\otimes \mathrm{V}_\nu^d)|_{\mu}$ denotes the subspace corresponding to the $\mathrm{V}_\mu^d\otimes\mathbb{C}^{c^\mu_{\lambda,\nu}}$ summand in the above decomposition.

Although computing the LR coefficients can be nontrivial in general~\cite[Appendix~A.8]{fulton2013representation}, the following simple necessary condition follows:
\begin{lemma}[Size-matching condition for the CG transformation]
    \label{lem:sizematching}
    If the staircases $\lambda$, $\nu$, and $\mu$ satisfy $c^{\mu}_{\lambda,\nu}>0$, then $|\mu|=|\lambda|+|\nu|$.
\end{lemma}
\begin{proof}
    For $\kappa\in\{\lambda,\nu,\mu\}$, let $\varphi_\kappa$ denote the group action of $\mathrm{U}(d)$ on $\mathrm{V}_\kappa^d$.
    If $c^\mu_{\lambda,\nu}>0$, the CG decomposition implies that there exists a nonzero intertwiner
    $T:\mathrm{V}_\lambda^d\otimes\mathrm{V}_\nu^d\to
    \mathrm{V}_\mu^d\otimes\mathbb{C}^{c^\mu_{\lambda,\nu}}$
    such that
    \begin{align}
        T(\varphi_\lambda(U)\otimes\varphi_\nu(U))
        =
        (\varphi_\mu(U)\otimes I_{c^\mu_{\lambda,\nu}})T
    \end{align}
    for all $U\in\mathrm{U}(d)$.
    Consider $U=zI_d$ for $z\in\mathrm{U}(1)$.
    By Eq.~\eqref{eq:z_indentity_cond},
    \begin{align}
        z^{|\lambda|+|\nu|}T=z^{|\mu|}T
    \end{align}
    holds for all $z\in\mathrm{U}(1)$.
    As $T\neq 0$, this implies $|\mu|=|\lambda|+|\nu|$, completing the proof.
\end{proof}

\section{Lower bound}
We prove that $N=\Omega(Mrd/\varepsilon)$ copies are necessary for cloning $M$ additional copies of an arbitrary $d$-dimensional state of rank at most $r$ with fidelity at least $1-\varepsilon$.
The key step in the proof is to show that it suffices to consider unitary-covariant and permutation-invariant cloning maps, and then decompose these maps into simpler component maps.
Here, a cloning map $\mathcal{T}:\mathrm{End}((\mathbb{C}^d)^{\otimes N})\to\mathrm{End}((\mathbb{C}^d)^{\otimes (N+M)})$ is called unitary-covariant if
\begin{align}
    \mathcal{T}(U^{\otimes N}XU^{\dagger\otimes N})
    =
    U^{\otimes (N+M)}\mathcal{T}(X)U^{\dagger\otimes (N+M)}
\end{align}
for all $U\in\mathrm{U}(d)$, and permutation-invariant if
\begin{align}
    \mathcal{T}(\varphi_{\mathrm{S}_N}(\sigma)X\varphi_{\mathrm{S}_N}(\sigma)^\dagger)=\mathcal{T}(X),\qquad\varphi_{\mathrm{S}_{N+M}}(\pi)\mathcal{T}(X)\varphi_{\mathrm{S}_{N+M}}(\pi)^\dagger
    =\mathcal{T}(X)
\end{align}
for all $\sigma\in\mathrm{S}_N$ and $\pi\in\mathrm{S}_{N+M}$.
We first introduce the necessary lemmas for this step and then proceed to the main proof.

We begin by showing that it suffices to consider cloning maps satisfying unitary covariance and permutation invariance.
\begin{lemma}[Symmetry reduction]\label{lem:symmetry-reduction}
    Given a cloning map $\mathcal{T}:\mathrm{End}((\mathbb{C}^d)^{\otimes N})\to\mathrm{End}((\mathbb{C}^d)^{\otimes (N+M)})$, there exists a unitary-covariant and permutation-invariant map $\bar{\mathcal{T}}$ whose worst-case cloning fidelity over states of rank at most $r$ is no smaller, i.e.,
    \begin{align}
        \inf_{\rho:\mathrm{rank}(\rho) \le r} F(\rho^{\otimes (N+M)}, \mathcal{T}(\rho^{\otimes N}))
        \leq 
        \inf_{\rho:\mathrm{rank}(\rho) \le r} F(\rho^{\otimes (N+M)}, \bar{\mathcal{T}}(\rho^{\otimes N})).
    \end{align}
\end{lemma}
\begin{proof}
    Let
    \begin{align}
        \bar{\mathcal{T}}_{\mathrm{U}(d)}(X)
        \coleq
        \int 
        U^{\dagger\otimes (N+M)}
        \mathcal{T}\!\left(U^{\otimes N}XU^{\dagger\otimes N}\right)
        U^{\otimes (N+M)}
        dU
    \end{align}
    be the unitary symmetrization of a cloning map $\mathcal{T}$, which is a well-defined unitary-covariant CPTP map.  
    Here, the integral is over the Haar measure on $\mathrm{U}(d)$.
    Using the unitary invariance of the set of states of rank at most $r$, we have
    \begin{align}
        \inf_{\rho:\mathrm{rank}(\rho) \le r} F(\rho^{\otimes (N+M)}, \mathcal{T}(\rho^{\otimes N}))
        &=\int\inf_{\rho:\mathrm{rank}(\rho) \le r} F((U\rho U^\dagger)^{\otimes (N+M)}, \mathcal{T}((U\rho U^\dagger)^{\otimes N}))dU\\
        &\leq\inf_{\rho:\mathrm{rank}(\rho) \le r}\int F((U\rho U^\dagger)^{\otimes (N+M)}, \mathcal{T}((U\rho U^\dagger)^{\otimes N}))dU\\
        &=\inf_{\rho:\mathrm{rank}(\rho) \le r}\int F(\rho^{\otimes (N+M)}, U^{\dagger\otimes (N+M)}\mathcal{T}((U\rho U^\dagger)^{\otimes N})U^{\otimes (N+M)})dU\\
        &\leq \inf_{\rho:\mathrm{rank}(\rho) \le r} F\left(\rho^{\otimes (N+M)}, \int U^{\dagger\otimes (N+M)}\mathcal{T}((U\rho U^\dagger)^{\otimes N})U^{\otimes (N+M)}dU\right)\\
        &=\inf_{\rho:\mathrm{rank}(\rho) \le r} F(\rho^{\otimes (N+M)}, \bar{\mathcal{T}}_{\mathrm{U}(d)}(\rho^{\otimes N})),
    \end{align}
    where the fourth line follows from the concavity of fidelity.
    Thus, unitary symmetrization does not decrease the worst-case fidelity.

    We similarly obtain permutation invariance by averaging over the permutation actions.
    Let
    \begin{align}
        \bar{\mathcal{T}}_{\mathrm{S}_{N+M}}(X) 
        \coleq \frac{1}{|\mathrm{S}_{N+M}|\cdot|\mathrm{S}_{N}|}\sum_{\pi\in \mathrm{S}_{N+M}} \sum_{\sigma\in \mathrm{S}_{N}} \varphi_{\mathrm{S}_{N+M}}(\pi)\mathcal{T}(\varphi_{\mathrm{S}_{N}}(\sigma) X \varphi_{\mathrm{S}_{N}}(\sigma)^\dagger)\varphi_{\mathrm{S}_{N+M}}(\pi)^\dagger
    \end{align}
    be the permutation symmetrization of a cloning map $\mathcal{T}$, which is also a well-defined permutation-invariant CPTP map.
    Then, we have
    \begin{align}
        \inf_{\rho:\mathrm{rank}(\rho) \le r} F(\rho^{\otimes (N+M)}, \mathcal{T}(\rho^{\otimes N}))
        &=\inf_{\rho:\mathrm{rank}(\rho) \le r} \frac{1}{|\mathrm{S}_{N+M}|}\sum_{\pi\in \mathrm{S}_{N+M}} F(\rho^{\otimes (N+M)}, \varphi_{\mathrm{S}_{N+M}}(\pi)\mathcal{T}(\rho^{\otimes N})\varphi_{\mathrm{S}_{N+M}}(\pi)^\dagger)\\
        &\leq \inf_{\rho:\mathrm{rank}(\rho) \le r} F\left(\rho^{\otimes (N+M)}, \frac{1}{|\mathrm{S}_{N+M}|}\sum_{\pi\in \mathrm{S}_{N+M}} \varphi_{\mathrm{S}_{N+M}}(\pi)\mathcal{T}(\rho^{\otimes N})\varphi_{\mathrm{S}_{N+M}}(\pi)^\dagger\right)\\
        &=\inf_{\rho:\mathrm{rank}(\rho) \le r} F(\rho^{\otimes (N+M)}, \bar{\mathcal{T}}_{\mathrm{S}_{N+M}}(\rho^{\otimes N}))
    \end{align}
    by the concavity of fidelity.
    Hence, permutation symmetrization also does not decrease the worst-case cloning fidelity.
    
    Finally, we consider the worst-case cloning fidelity of the composite map 
    \begin{align}
        \bar{\mathcal{T}}(X)\coleq\frac{1}{|\mathrm{S}_{N+M}|\cdot|\mathrm{S}_{N}|}\sum_{\pi\in \mathrm{S}_{N+M}} \sum_{\sigma\in \mathrm{S}_{N}} \varphi_{\mathrm{S}_{N+M}}(\pi)\bar{\mathcal{T}}_{\mathrm{U}(d)}(\varphi_{\mathrm{S}_{N}}(\sigma) X \varphi_{\mathrm{S}_{N}}(\sigma)^\dagger)\varphi_{\mathrm{S}_{N+M}}(\pi)^\dagger.
    \end{align}
    This map is permutation-invariant by construction. 
    Moreover, both input and output permutation actions commute with the corresponding collective unitary actions, so the permutation symmetrization preserves unitary covariance.
    Hence, $\bar{\mathcal{T}}$ is both unitary-covariant and permutation-invariant. Therefore, the worst-case cloning fidelity of $\bar{\mathcal{T}}$ over states of rank at most $r$ is no smaller than that of $\mathcal{T}$.
\end{proof}

We next introduce a lemma implying that cloning maps with these symmetries admit a structured decomposition.
\begin{lemma}[Decomposition of unitary-covariant and permutation-invariant maps~\cite{manvcinska2025classification}]
    \label{lem:decomposition}
    Let $\mathcal{T}:\mathrm{End}((\mathbb{C}^d)^{\otimes N})\to\mathrm{End}((\mathbb{C}^d)^{\otimes (N+M)})$ be a unitary-covariant and permutation-invariant CPTP map.
    Any permutation-invariant input operator can be written as
    \begin{align}
        X\cong\bigoplus_{\lambda\vdash_d N} \frac{{I}_{\mathrm{Sp}_\lambda}}{\dim(\mathrm{Sp}_\lambda)}\otimes X_{\lambda}, \quad X_\lambda \in \mathrm{End}(\mathrm{V}_\lambda^d).
    \end{align}
    For such an input, the output decomposes as
    \begin{align}
        \mathcal{T}(X)
        &\cong\sum_{\lambda\vdash_d N}\mathcal{T}_{\lambda}(X_\lambda),\\
        \mathcal{T}_{\lambda}
        &\in\mathrm{Conv}
        (
        \{
        {\mathcal{T}}_{\lambda\to\mu}^{\nu,\psi}:\mu\vdash_d N+M,\ket{\psi}\in\mathbb{C}^{c_{\lambda, \nu}^\mu}, \|\psi\|=1
        \}
        ),
    \end{align}
    where $\mathrm{Conv}$ denotes the convex hull and $\nu$ ranges over length-$d$ staircases.
    The component map ${\mathcal{T}}_{\lambda\to\mu}^{\nu,\psi}:\mathrm{End}(\mathrm{V}_\lambda^d)\to\mathrm{End}(\mathrm{Sp}_\mu\otimes \mathrm{V}_\mu^d)$ is given by
    \begin{align}
        \label{eq:component-map}
        {\mathcal{T}}_{\lambda\to\mu}^{\nu,\psi}(X_\lambda)
        &=\frac{\dim (\mathrm{V}_\lambda^d)}{\dim( \mathrm{V}_\mu^d)}\frac{I_{\mathrm{Sp}_\mu}}{\dim(\mathrm{Sp}_\mu)}\otimes U_{\mu,\psi\to \lambda,\nu}^{\mathrm{CG}\dagger}(X_\lambda \otimes I_{\mathrm{V}_\nu^d})U_{\mu,\psi\to \lambda,\nu}^{\mathrm{CG}},
    \end{align}
    where the isometry $U_{\mu,\psi\to \lambda,\nu}^{\mathrm{CG}}:\mathrm{V}_{\mu}^d\to\mathrm{V}_{\lambda}^d\otimes \mathrm{V}_{\nu}^d$ is defined by $U_{\mu,\psi\to \lambda,\nu}^{\mathrm{CG}}\coleq U_{\mu\to \lambda,\nu}^{\mathrm{CG}}(I_{\mathrm{V}^d_\mu}\otimes \ket{\psi})$.
\end{lemma}
\noindent Note that the original result in Ref.~\cite{manvcinska2025classification} is stated for irreps of $\mathrm{SU}(d)$ rather than $\mathrm{U}(d)$.
The extension from $\mathrm{SU}(d)$ to $\mathrm{U}(d)$ follows from the fact that $\mathrm{U}(d)$ irreps whose labels differ by a common integer shift restrict to the same $\mathrm{SU}(d)$ irrep.
More precisely, for each $\mathrm{SU}(d)$ irrep appearing in the original decomposition, one can choose a corresponding $\mathrm{U}(d)$ lift so that the actions of the central $\mathrm{U}(1)$ subgroup agree on the two sides of the intertwining isometry, which then becomes a $\mathrm{U}(d)$ intertwiner.
Consequently, the same form of the component maps applies to $\mathrm{U}(d)$, yielding Lemma~\ref{lem:decomposition}.

We are now ready to prove the lower bound.
\begin{theorem}[Lower bound]\label{thm:lower-bound}
    Let $1\leq r\leq d$ with $d\geq 2$ and $0<\varepsilon<1/30$. 
    Any cloning map that achieves worst-case fidelity at least $1-\varepsilon$ for cloning $M$ additional copies of any $d$-dimensional quantum state of rank at most $r$ requires $N=\Omega(Mrd/\varepsilon)$ input copies.
\end{theorem}
\begin{proof}
    We prove that any $N\to N+M$ cloning map $\mathcal{T}$ with
    \begin{align}
        \label{eq:epassumption}
        \inf_{\rho:\mathrm{rank}(\rho) \le r} F(\rho^{\otimes (N+M)}, \mathcal{T}(\rho^{\otimes N}))\geq 1-\varepsilon
    \end{align} 
    must satisfy $N=\Omega(Mrd/\varepsilon)$.
    We establish this by deriving an upper bound on the worst-case cloning fidelity and comparing it with the assumed lower bound $1-\varepsilon$.

    By Lemma~\ref{lem:symmetry-reduction}, we may assume that the cloning map $\mathcal{T}$ is unitary-covariant and permutation-invariant.
    To upper-bound the worst-case cloning fidelity for such a map, we consider the fidelity for a particular hard instance.
    Fix some integer $l$ with $1 \le l \le r$ and choose the rank-$l$ maximally mixed state $\rho_l \coleq (I_l\oplus 0_{d-l})/l=\Pi_l/l$ as a hard instance, where $I_l$ denotes the $ l$-dimensional identity operator, $0_{d-l}$ denotes the $d-l$-dimensional zero operator, and $\Pi_l$ denotes the projector onto the subspace spanned by the first $l$ basis vectors.
    The fidelity is upper bounded as follows:
    \begin{align}
        F(\rho_l^{\otimes (N+M)},\mathcal{T}(\rho_l^{\otimes N}))
        &=\left(\tr\left(\sqrt{\sqrt{\rho_l^{\otimes (N+M)}}\mathcal{T}(\rho_l^{\otimes N})\sqrt{\rho_l^{\otimes (N+M)}}}\right)\right)^2\\
        &=\left(\tr\left(\sqrt{\frac{1}{l^{N+M}}\Pi_l^{\otimes (N+M)}\mathcal{T}(\rho_l^{\otimes N})\Pi_l^{\otimes (N+M)}}\right)\right)^2\\
        &\leq\tr (\Pi_l^{\otimes (N+M)}\mathcal{T}(\rho_l^{\otimes N}))
    \end{align}
    where the last line follows from the Cauchy-Schwarz inequality.

    To further bound the RHS, we decompose the cloned state $\mathcal{T}(\rho_l^{\otimes N})$ in the Schur basis.
    Since the cloning map $\mathcal{T}$ is unitary-covariant and permutation-invariant, Lemma~\ref{lem:decomposition} gives the decomposition
    \begin{align}
        \mathcal{T}(X)\cong\sum_{\lambda \vdash_d N} \mathcal{T}_{\lambda}(X_\lambda),
    \end{align}
    where $\mathcal{T}_\lambda$ is a convex combination of the component maps ${\mathcal{T}}_{\lambda\to\mu}^{\nu,\psi}$ given in Eq.~\eqref{eq:component-map}, and $X_\lambda$ is the $\mathrm{V}^d_\lambda$-sector of the input $X$.
    The input state $\rho_l^{\otimes N}$ and the projector $\Pi_l^{\otimes(N+M)}$ can be decomposed in the Schur basis as
    \begin{align}
        \rho_l^{\otimes N}
        &\cong
        \bigoplus_{\lambda\vdash_l N} p_l(\lambda)\frac{{I}_{\mathrm{Sp}_\lambda}}{\dim(\mathrm{Sp}_\lambda)}\otimes \frac{P_\lambda^{d, l}}{\dim(\mathrm{V}_\lambda^l)},\\
        \Pi_l^{\otimes(N+M)}
        &\cong
        \bigoplus_{\mu\vdash_l N+M} {I}_{\mathrm{Sp}_\mu}\otimes {P_\mu^{d,l}},
    \end{align}
    where $P_\lambda^{d,l}\in\mathrm{End}(\mathrm{V}_\lambda^d)$ denotes the projector onto the $\mathrm{V}_\lambda^l$-subspace, and the occupancy probability $p_l(\lambda)$ is defined by
    \begin{align}
        p_l(\lambda)\coleq\frac{\dim(\mathrm{Sp}_\lambda)\dim(\mathrm{V}_\lambda^l)}{l^N}.
    \end{align}
    Given this, the RHS can be rewritten as
    \begin{align}
        \mathrm{RHS}
        &=\tr \left(\left(\sum_{\lambda\vdash_l N}\mathcal{T}_{\lambda}\left(p_l(\lambda)\frac{P_\lambda^{d, l}}{\dim(\mathrm{V}_\lambda^l)}\right)\right)\left(\bigoplus_{\mu\vdash_l N+M} {I}_{\mathrm{Sp}_\mu}\otimes {P_\mu^{d, l}}\right)\right)\\
        &=\tr \left(\left(\sum_{\lambda\vdash_l N}\bigoplus_{\mu\vdash_l N+M}\sum_{\nu,\psi}\alpha_{\mu,\nu,\psi} \mathcal{T}_{\lambda\to\mu}^{\nu,\psi}\left(p_l(\lambda)\frac{P_\lambda^{d, l}}{\dim(\mathrm{V}_\lambda^l)}\right)\right)\left(\bigoplus_{\mu\vdash_l N+M} {I}_{\mathrm{Sp}_\mu}\otimes {P_\mu^{d, l}}\right)\right)\\
        &=\sum_{\lambda\vdash_l N}p_l(\lambda)
        \sum_{\mu\vdash_l N+M}
        \sum_{\nu,\psi}\alpha_{\mu,\nu,\psi} \tr \left( \mathcal{T}_{\lambda\to\mu}^{\nu,\psi}\left(\frac{P_\lambda^{d, l}}{\dim(\mathrm{V}_\lambda^l)}\right)( {I}_{\mathrm{Sp}_\mu}\otimes {P_\mu^{d, l}})\right)\\
        &=\sum_{\lambda\vdash_l N}p_l(\lambda)
        \sum_{\mu\vdash_l N+M}
        \sum_{\nu,\psi}\alpha_{\mu,\nu,\psi} \frac{\dim(\mathrm{V}_\lambda^d)}{\dim(\mathrm{V}_\mu^d)\dim(\mathrm{V}_\lambda^l)} \tr (U_{\mu,\psi\to \lambda,\nu}^{\mathrm{CG}\dagger}(P_\lambda^{d, l}\otimes I_{\mathrm{V}_\nu^d})U_{\mu,\psi\to \lambda,\nu}^{\mathrm{CG}}P_\mu^{d, l}),\label{eq:alpha_sum}
    \end{align}
    where $\alpha_{\mu,\nu,\psi}\geq 0$ are coefficients of the convex combination, satisfying $\sum_{\mu\vdash_d N+M,\nu,\psi}\alpha_{\mu,\nu,\psi}=1$ for each $\lambda$-sector.
    
    To bound this expression, we use the following lemma, whose proof is the most technical part of the argument:
    \begin{lemma}[Upper bound on cloning fidelity for a component map]
        \label{lem:BCI-general}
        For partitions $\lambda\vdash_l N$ and $\mu\vdash_l N+M$ with $1\le l\le d$, a length-$d$ staircase $\nu$, and an isometric $\mathrm{U}(d)$-intertwiner $U:\mathrm{V}_\mu^d\to \mathrm{V}_\lambda^d\otimes \mathrm{V}_\nu^d$, the following inequality holds:
        \begin{align}
            \frac{\dim(\mathrm{V}_\lambda^d)}{\dim(\mathrm{V}_\mu^d)\dim(\mathrm{V}_\lambda^l)} \tr (U^\dagger(P_\lambda^{d, l}\otimes I_{\mathrm{V}_\nu^d})U P_\mu^{d, l})\leq\prod_{i=l+1}^d\frac{\lambda_1+i+M-2}{\lambda_1+i+2M-2}.
        \end{align}
    \end{lemma}
    \begin{proof}
        See Sec.~\ref{sec:l7proof}.
    \end{proof}
    Applying the lemma to Eq.~\eqref{eq:alpha_sum} and using the normalization of the coefficients $\alpha_{\mu,\nu,\psi}$, we obtain
    \begin{align}
        \mathrm{RHS}\leq \sum_{\lambda\vdash_l N} p_l(\lambda) \prod_{i=l+1}^d\frac{\lambda_1+i+M-2}{\lambda_1+i+2M-2}.
    \end{align}
    The product term is bounded by
    \begin{align}
        \prod_{i=l+1}^d\frac{\lambda_1+i+M-2}{\lambda_1+i+2M-2}
        &= \prod_{i=l+1}^d\left(1-\frac{M}{\lambda_1+i+2M-2}\right)\\
        &\leq \left(1-\frac{M}{\lambda_1+d+2M-2}\right)^{d-l}\\
        &\leq 1-\frac{M(d-l)}{\lambda_1+d+2M-2+M(d-l)},
    \end{align}
    where the last line follows from $(1-x)^n\leq (1+nx)^{-1}$.
    Thus, we can bound the RHS of Eq.~\eqref{eq:alpha_sum} by
    \begin{align}
        \mathrm{RHS}
        &\leq 1-\sum_{\lambda\vdash_l N} p_l(\lambda)\frac{M(d-l)}{\lambda_1+d+2M-2+M(d-l)}\\
        &\leq  1-\frac{M(d-l)}{\sum_{\lambda\vdash_l N}p_l(\lambda)\lambda_1+d+2M-2+M(d-l)}
    \end{align}
    where the last line follows from Jensen's inequality applied to the convex function $f(x) = 1/(x+c)$ for a constant $c$.

    It remains to bound the expectation 
    $\sum_{\lambda\vdash_l N}p_l(\lambda)\lambda_1$.
    Theorem~5.2 in Ref.~\cite{o2016efficient} states that
    \begin{align}
        \sum_{\lambda\vdash_l N}p_l(\lambda)\lambda_1
        \leq
        \frac{N}{l}+2\sqrt{N},
    \end{align}
    which yields the following upper bound on the cloning fidelity:
    \begin{align}
        F(\rho_l^{\otimes (N+M)},\mathcal{T}(\rho_l^{\otimes N}))
        &\leq 
        1-\frac{M(d-l)}{{N}/{l}+2\sqrt{N}+d+2M-2+M(d-l)}.
    \end{align}

    We turn the above upper bound into a lower bound on $N$ by comparing it with the assumed worst-case cloning fidelity $1-\varepsilon$ and choosing $l$ appropriately.
    Combining the upper bound with the fidelity assumption, we must have
    \begin{align}
        \varepsilon
        &\geq
        \frac{M(d-l)}{{N}/{l}+2\sqrt{N}+d+2M-2+M(d-l)}
    \end{align}
    for every $1\leq l\leq r$.
    Choose $l=\lceil r/2\rceil$. 
    From $1\leq r\leq d$ and $d\geq 2$, we have $l=\lceil r/2\rceil\leq \lceil d/2\rceil\leq {2d}/{3}$, which gives $d\leq 3(d-l)$ and $l\leq 2(d-l)$. 
    Suppose $N<4l^2$. 
    Then we have $N/l<4l$ and $2\sqrt{N}<4l$, leading to
    \begin{align}
        \frac{M(d-l)}{{N}/{l}+2\sqrt{N}+d+2M-2+M(d-l)}
        &>\frac{M(d-l)}{8l+d+2M-2+M(d-l)}\\
        &>\frac{M(d-l)}{19(d-l)+3M(d-l)}\\
        &\geq\frac{1}{22},
    \end{align}
    which contradicts the assumption $\varepsilon<1/30$. 
    Thus, $N\geq 4l^2$.
    This yields $2\sqrt{N}\leq N/l$, leading to
    \begin{align}
        \varepsilon
        &\geq
        \frac{M(d-l)}{{N}/{l}+2\sqrt{N}+d+2M-2+M(d-l)}\\
        &\geq
        \frac{M(d-l)}{{2N}/{l}+d+3M(d-l)}\\
        &\geq
        \frac{Ml(d-l)}{2N+ld+3Ml(d-l)}.
    \end{align}
    Using $r/2\leq l\leq r$ and $d/3\leq d-l\leq d$, 
    \begin{align}
        \varepsilon \geq \frac{Mrd/6}{2N+rd+3Mrd} \geq \frac{Mrd}{12N+24Mrd}.
    \end{align}
    The assumption $\varepsilon<1/30$ then implies
    \begin{align}
        12\varepsilon N \geq Mrd(1-24\varepsilon) \geq
        \frac{1}{5}Mrd,
    \end{align}
    and hence yields
    \begin{align}
        N
        \geq
        \frac{1}{60}\frac{Mrd}{\varepsilon}.
    \end{align}
    Therefore, $N=\Omega(Mrd/\varepsilon)$, which completes the proof.

\end{proof}

\section{Upper bound}

We prove that $N=O(Mrd/\varepsilon)$ input copies are sufficient to clone $M$ additional copies of any $d$-dimensional state of rank at most $r$ with fidelity at least $1-\varepsilon$.
The construction is based on the random purification channel~\cite{pelecanos2025mixed, tang2025conjugate}, which lifts a tensor power of a mixed state to a mixture of identical pure-state purifications.

Let $\rho$ be a state of rank at most $r$, and fix a purification $\ket{\Psi_{\rho}}\in\mathbb{C}^d\otimes\mathbb{C}^r$.
For $U\in\mathrm{U}(r)$, define the rotated purification
\begin{align}
    \ket{\Psi_{\rho,U}}\coleq (I\otimes U)\ket{\Psi_{\rho}}.
\end{align}
The random purification channel $\Phi$ is then defined by
\begin{align}
    \Phi(\rho^{\otimes N})
    &=\E_U\ketbra{\Psi_{\rho,U}}{\Psi_{\rho,U}}^{\otimes N},
\end{align}
where the expectation is taken over the Haar measure on $\mathrm{U}(r)$.
As shown in Refs.~\cite{pelecanos2025mixed, tang2025conjugate}, this prescription defines a completely positive and trace-preserving map acting on the input state $\rho^{\otimes N}$.
Moreover, its output is supported on the symmetric subspace of $(\mathbb{C}^{rd})^{\otimes N}$.

We now define the random-purification cloning map $\mathcal{T}_{\mathrm{RP}}$.
It first applies the random purification channel, then applies the optimal pure-state cloning map in dimension $rd$, and finally traces out the purification registers:
\begin{align}
    \mathcal{T}_{\mathrm{RP}}(\rho^{\otimes N})
    \coleq
    \tr_R\mathcal{T}_{\mathrm{Pure}}(\Phi(\rho^{\otimes N})).
\end{align}
Here, $\tr_R$ denotes the partial trace over the purification registers.
The pure-state cloner used above is the optimal $N\to N+M$ pure-state cloning map on the $rd$-dimensional purified Hilbert space:
\begin{align}
    \mathcal{T}_{\mathrm{Pure}}(X)
    \coleq
    \frac{\dim(\mathrm{Sym}_N^{rd})}
         {\dim(\mathrm{Sym}_{N+M}^{rd})}
    \Pi_{\mathrm{Sym}_{N+M}^{rd}}
    \left(
        X\otimes I_{rd}^{\otimes M}
    \right)
    \Pi_{\mathrm{Sym}_{N+M}^{rd}},
\end{align}
where $X$ is an operator supported on $\mathrm{Sym}_N^{rd}$, $I_{rd}$ is the identity operator on $\mathbb C^{rd}$, and $\Pi_{\mathrm{Sym}_{N+M}^{rd}}$ is the projector onto the symmetric subspace of $(\mathbb C^{rd})^{\otimes(N+M)}$.

This gives the following theorem:

\begin{theorem}[Upper bound]
    The random-purification cloning map $\mathcal{T}_{\mathrm{RP}}$ achieves worst-case fidelity at least $1-\varepsilon$ for cloning $M$ additional copies of any $d$-dimensional state of rank at most $r$ using $N=O(Mrd/\varepsilon)$ input copies.
\end{theorem}
\begin{proof}
    We have
    \begin{align}
        F(\rho^{\otimes (N+M)}, \mathcal{T}_{\mathrm{RP}}(\rho^{\otimes N}))
        &=F(\tr_R \E_U\ketbra{\Psi_{\rho,U}}{\Psi_{\rho,U}}^{\otimes (N+M)}, \tr_R\mathcal{T}_{\mathrm{Pure}}(\E_U\ketbra{\Psi_{\rho,U}}{\Psi_{\rho,U}}^{\otimes N}))\\
        &\geq F(\E_U\ketbra{\Psi_{\rho,U}}{\Psi_{\rho,U}}^{\otimes (N+M)}, \mathcal{T}_{\mathrm{Pure}}(\E_U\ketbra{\Psi_{\rho,U}}{\Psi_{\rho,U}}^{\otimes N}))\\
        &=\frac{\dim(\mathrm{Sym}_N^{rd})}{\dim(\mathrm{Sym}_{N+M}^{rd})} F(\E_U\ketbra{\Psi_{\rho,U}}{\Psi_{\rho,U}}^{\otimes (N+M)},\E_U\ketbra{\Psi_{\rho,U}}{\Psi_{\rho,U}}^{\otimes N}\otimes I_{rd}^{\otimes M})\\
        & \geq \frac{\dim(\mathrm{Sym}_N^{rd})}{\dim(\mathrm{Sym}_{N+M}^{rd})}\left(\E_U\sqrt{F(\ketbra{\Psi_{\rho,U}}{\Psi_{\rho,U}}^{\otimes (N+M)},\ketbra{\Psi_{\rho,U}}{\Psi_{\rho,U}}^{\otimes N}\otimes I_{rd}^{\otimes M})}\right)^2\\
        &=\frac{\dim(\mathrm{Sym}_N^{rd})}{\dim(\mathrm{Sym}_{N+M}^{rd})}\\
        &=\frac{\binom{rd+N-1}{N}}{\binom{rd+N+M-1}{N+M}},
    \end{align}
    where the second line follows from the monotonicity of fidelity under the CPTP map $\tr_R$ and the fourth line follows from the joint concavity of root fidelity.
    The RHS can be bounded by
    \begin{align}
        \mathrm{RHS}
        &=
        \prod_{i=1}^{M}
        \frac{N+i}{N+rd+i-1}\\
        &=
        \prod_{i=1}^{M}
        \left(
            1-\frac{rd-1}{N+rd+i-1}
        \right)\\
        &\geq
        1-\sum_{i=1}^{M}
        \frac{rd-1}{N+rd+i-1}\\
        &\geq
        1-\frac{Mrd}{N},
    \end{align}
    where the third line follows from $\prod_{i=1}^{M}(1-x_i)\geq 1-\sum_{i=1}^{M}x_i$ for $0\leq x_i \leq 1$.
    Hence, $N\geq Mrd/\varepsilon$ is sufficient to guarantee fidelity at least $1-\varepsilon$, completing the proof.
\end{proof}

\section{Relation with tomography}

We prove that high-fidelity tomography can be coherently converted into high-fidelity cloning.
We note that this is a generalization of Corollary 1.3 in Ref.~\cite{fefferman2025hardness}, where a similar result was proved for pure states.
The generalization is achieved by considering the cloning fidelity for purifications of the mixed states and relating it to the mixed state cloning fidelity.

\begin{theorem}[Tomography to cloning]
    For any class of quantum states, suppose there exists a tomography algorithm that estimates the state with fidelity at least $1-\varepsilon$ using $N$ samples with $\delta=O(\varepsilon)$ failure probability.
    Then there exists an $N\to N+1$ cloning map that produces an additional copy with fidelity at least $1-c\varepsilon$ for some constant $c$, using the same $N$ input copies.
\end{theorem}

\begin{proof}
    We construct the cloning map by sequentially performing a coherent tomography measurement, a controlled preparation of an additional copy, and an uncomputation of the coherent tomography measurement.

    Let $\{K_x\}$ be Kraus operators associated with the tomography POVM, satisfying $\sum_x K_x^\dagger K_x=I_d^{\otimes N}$.
    On input $\rho^{\otimes N}$, outcome $x$ occurs with probability
    \begin{align}
        p_x=\tr(K_x\rho^{\otimes N}K_x^\dagger),
    \end{align}
    and produces an estimate $\hat{\rho}_x$ satisfying
    \begin{align}
        \Pr_{x\sim \{p_x\}_x}\left[F(\rho,\hat{\rho}_x)< 1-\varepsilon\right]\leq \delta
    \end{align}
    for every $\rho$ in the class.
    The measurement can be implemented coherently by the isometry
    \begin{align}
        V=\sum_x K_x\otimes \ket{x}.
    \end{align}
    We define the cloning map $\mathcal{T}$ as follows: apply $V$ to the input, prepare $\hat{\rho}_x$ on an additional register controlled on the coherent outcome register $\ket{x}$, and then apply $V^\dagger$ to uncompute the tomography procedure; see Fig.~\ref{fig:sup-circuit}.

    \begin{figure}[h!]
        \centering
        \includegraphics{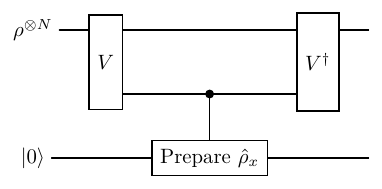}
        \caption{Cloning map $\mathcal{T}$ constructed from a tomography algorithm.}
        \label{fig:sup-circuit}
    \end{figure}
    \noindent Note that the use of the uncomputation $V^\dagger$ is well justified, since the isometry $V$ can always be extended to a unitary on a larger Hilbert space, whose inverse implements the required uncomputation.

    We now lower bound the fidelity between the output of this map and the ideal state $\rho^{\otimes (N+1)}$.
    By Uhlmann's theorem, it suffices to establish the same fidelity bound for suitable purifications of $\rho^{\otimes (N+1)}$ and $\mathcal{T}(\rho^{\otimes N})$.
    Choose a purification $\ket{\Psi_{\rho^{\otimes N}}}$ of the input state and define the coherent implementation of the measurement on the purified system by
    \begin{align}
        W=\sum_x K_x\otimes I\otimes\ket{x}.
    \end{align}
    The purified version of the above cloning procedure, denoted by $\mathcal{T}_{\mathrm{Purified}}$, applies $W$, prepares a purification $\ket{\Psi_{\hat{\rho}_x}}$ of the estimated state on an additional system, and then applies $W^\dagger$; see Fig.~\ref{fig:sup-circuit-purified}.

    \begin{figure}[h!]
    \centering
    \includegraphics{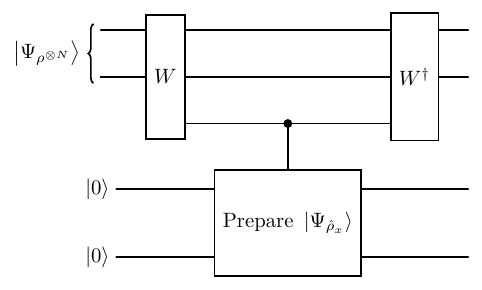}
    \caption{Purified cloning map $\mathcal{T}_{\mathrm{Purified}}$ constructed from a tomography algorithm.}
    \label{fig:sup-circuit-purified}
    \end{figure}

    The output state $\mathcal{T}_{\mathrm{Purified}}(\ket{\Psi_{\rho^{\otimes N}}})$ is a purification of $\mathcal{T}(\rho^{\otimes N})$.
    Thus, it suffices to show that
    $F(\ket{\Psi_{{\rho}^{\otimes (N+1)}}}, \mathcal{T}_{\mathrm{Purified}}(\ket{\Psi_{\rho^{\otimes N}}}))\geq 1-c\varepsilon$
    holds for all $\rho$.
    The output state is written as
    \begin{align}
        \mathcal{T}_{\mathrm{Purified}}(\ket{\Psi_{\rho^{\otimes N}}})
        &=(W^\dagger\otimes I)\left(\sum_x(K_x\otimes I)\ket{\Psi_{\rho^{\otimes N}}}\otimes \ket{x}\otimes\ket{\Psi_{\hat{\rho}_x}}\right),
    \end{align}
    while a purification of the ideal output state $\rho^{\otimes (N+1)}$ can be written as
    \begin{align}
        \ket{\Psi_{\rho^{\otimes (N+1)}}}=(W^\dagger\otimes I)\left(\sum_x(K_x\otimes I)\ket{\Psi_{\rho^{\otimes N}}}\otimes \ket{x}\otimes\ket{\Psi_\rho}\right).
    \end{align}
    Therefore, we obtain
    \begin{align}
        F(\rho^{\otimes (N+1)}, \mathcal{T}(\rho^{\otimes N}))
        &\geq F(\ket{\Psi_{{\rho}^{\otimes (N+1)}}}, \mathcal{T}_{\mathrm{Purified}}(\ket{\Psi_{\rho^{\otimes N}}}))\\
        &= F\left(\sum_x(K_x\otimes I)\ket{\Psi_{\rho^{\otimes N}}}\otimes \ket{x}\otimes\ket{\Psi_\rho}, \sum_x(K_x\otimes I)\ket{\Psi_{\rho^{\otimes N}}}\otimes \ket{x}\otimes\ket{\Psi_{\hat{\rho}_x}}\right)\\
        &=\left|\sum_x\tr(K_x\rho^{\otimes N}K_x^\dagger)\braket{\Psi_\rho}{\Psi_{\hat{\rho}_x}}\right|^2.
    \end{align}
    By choosing the canonical purifications,
    $\ket{\Psi_\rho}=\sqrt{d}(\sqrt{\rho}\otimes I)\ket{\Gamma}$ and $\ket{\Psi_{\hat{\rho}_x}}=\sqrt{d}(\sqrt{\hat{\rho}_x}\otimes I)\ket{\Gamma}$,
    where $\ket{\Gamma}$ is a maximally entangled state, we obtain
    \begin{align}
        \mathrm{RHS}
        &=\left|\sum_x\tr(K_x\rho^{\otimes N}K_x^\dagger)\tr(\sqrt{\rho}\sqrt{\hat{\rho}_x})\right|^2\\
        &\geq\left|\sum_x\tr(K_x\rho^{\otimes N}K_x^\dagger)F({\rho},{\hat{\rho}_x})\right|^2\\
        &\geq((1-\delta)(1-\varepsilon))^2
    \end{align}
    where the second line follows from Theorem 3 of Ref.~\cite{audenaert2012comparisons}.
    Thus, if $\delta\leq c'\varepsilon$ for some constant $c'$, the fidelity is at least $1-c\varepsilon$ for $c=2(c'+1)$, completing the proof.

\end{proof}

\section{Proof of technical lemmas}

We provide proofs of the technical lemmas needed to establish the main result in Theorem~\ref{thm:lower-bound}.
Since the logical dependencies among these proofs are somewhat complex, we first outline the organization of this section.
The section proceeds as follows.
We first prove Lemma~\ref{lem:BCI-general} in Sec.~\ref{sec:l7proof}, which provides the key inequality needed for Theorem~\ref{thm:lower-bound}.
Its proof proceeds by induction, with the induction step reduced to Lemma~\ref{lem:bci_one_step}.
We then prove Lemma~\ref{lem:bci_one_step} in Sec.~\ref{sec:l8proof} by reducing it to three auxiliary lemmas: Lemmas~\ref{lem:step1}, \ref{lem:step2}, and \ref{lem:step3}.
Finally, we prove these three lemmas in Secs.~\ref{sec:l9proof}, \ref{sec:l10proof}, and \ref{sec:l11proof}, respectively.
The overall structure of the proof of Theorem~\ref{thm:lower-bound} is illustrated in Fig.~\ref{fig:proof-diagram}.

\begin{figure}[hb]
    \centering
    \includegraphics{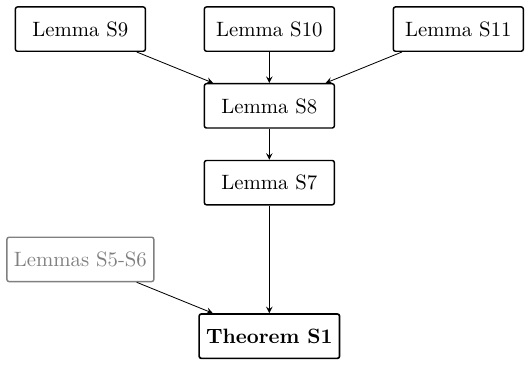}
    \caption{Proof structure for Theorem~\ref{thm:lower-bound}. An arrow indicates that the source result is used to prove the target result.}
    \label{fig:proof-diagram}
\end{figure}

We also summarize the notation used throughout this section in Table~\ref{tab:notation}.

\begin{table}[ht]
  \caption{Notation}
  \label{tab:notation}
  \centering

  \begin{tabular}{p{0.2\columnwidth} p{0.7\columnwidth}}
    \hline
    \textbf{Symbol} & \textbf{Description} \\
    \hline
    $\lambda,\mu,\nu, \kappa, \gamma$
    &
    Partitions $\lambda,\mu$ with $|\mu|=|\lambda|+M$; a length-$d$ staircase $\nu$ with $|\nu|=M$; and general length-$d$ and length-$(d-1)$ staircases $\kappa$ and $\gamma$, respectively, satisfying $\gamma\prec\kappa$.
    
    \\
    $\mathbb C_q$
    &
    The one-dimensional representation of $\mathrm U(1)$ defined by $z\mapsto z^q$.
    \\
    $J_{\gamma,\kappa}$
    &
    The inclusion isometry $J_{\gamma,\kappa}:\mathrm V_\gamma^{d-1}\to \mathrm V_\kappa^d$ for $\gamma\prec\kappa$.

     \\
    $J_\lambda$, $J_\mu$
    &
    Shorthand for $J_{\lambda,\lambda}$, $J_{\mu,\mu}$.
    \\
    
    $P^{d, d-1}_{\kappa,\gamma}$
    &
    The projector onto the $\mathrm{V}_\gamma^{d-1}$-subspace of $\mathrm{V}_\kappa^d$,
    $P^{d,d-1}_{\kappa,\gamma} \coleq J_{\gamma,\kappa} J_{\gamma, \kappa}^\dagger$.
    \\
    $P^{d,l}_{\kappa}$
    &
    The projector onto the $\mathrm{V}_\kappa^l$-subspace of $\mathrm{V}_\kappa^d$.
    \\
    $Q_{\kappa,q}$
    &
    The projector onto the weight-$q$ subspace of $\mathrm{V}_\kappa^d$,
    $Q_{\kappa,q}\coleq\sum_{\substack{\gamma\prec\kappa, 
    |\kappa|-|\gamma|=q}}P_{\kappa,\gamma}^{d,d-1}$, $q \in \mathbb{Z}$.
    \\

    $U$
    &
    An isometric $\mathrm U(d)$-intertwiner, $U : \mathrm{V}_{\mu}^d \to \mathrm{V}_\lambda^d \otimes \mathrm{V}_\nu^d$.
    \\
    $K_\gamma$, $c_\gamma$
    &
    $K_\gamma\coleq(J_\lambda^\dagger\otimes
    J_{\gamma,\nu}^\dagger)UJ_\mu$, $K_\gamma^\dagger K_\gamma= c_\gamma I_{\mathrm{V}_\mu^{d-1}}$.
    \\
    $B_q$
    &
    $B_q\coleq(Q_{\lambda,0}\otimes Q_{\nu,q})UQ_{\mu,q}$ for $q \in \mathbb{Z}$.
    \\
    $B_{0,\gamma}$
    &
    $B_{0,\gamma}\coleq (I_{\mathrm{V}_\lambda^d}\otimes
    P_{\nu,\gamma}^{d,d-1})B_0$ for $\gamma \prec \nu$ with $|\nu| - |\gamma| = 0$.
    \\
    $E_\kappa^{i,j}$
    &
    The infinitesimal action on $\mathrm{V}_\kappa^d$ corresponding to the standard matrix unit $|i\rangle\langle j|\in\mathfrak{gl}_d(\mathbb C)$.
    \\
    $e_s$
    &
    The length-$(d-1)$ tuple with all entries zero except for a one in the $s$-th position.
    \vspace{0.1em}
    \\
    \hline
  \end{tabular}
\end{table}

\subsection{Proof of Lemma~\ref{lem:BCI-general}}
\label{sec:l7proof}
We show that
    \begin{align}
    \label{eq:BCI-general}
        \frac{\dim(\mathrm{V}_\lambda^d)}{\dim(\mathrm{V}_\mu^d)\dim(\mathrm{V}_\lambda^l)} \tr (U^\dagger(P_\lambda^{d, l}\otimes I_{\mathrm{V}_\nu^d})U P_\mu^{d, l})\leq\prod_{i=l+1}^d\frac{\lambda_1+i+M-2}{\lambda_1+i+2M-2}
    \end{align}
holds for partitions $\lambda\vdash_l N$ and $\mu\vdash_l N+M$ with $1\le l\le d$, a length-$d$ staircase $\nu$, and an isometric $\mathrm{U}(d)$-intertwiner $U:\mathrm{V}_\mu^d\to \mathrm{V}_\lambda^d\otimes \mathrm{V}_\nu^d$.

We prove this by induction on $d-l$.
In the base case $d=l$, the product on the RHS is empty and hence equals one.
Thus, the inequality follows from
\begin{align}
     \frac{1}{\dim(\mathrm{V}_\mu^d)} \tr (U^\dagger(P_\lambda^{d, d}\otimes I_{\mathrm{V}_\nu^d})U P_\mu^{d, d})
    &=\frac{1}{\dim(\mathrm{V}_\mu^d)}\tr(U^\dagger(I_{\mathrm{V}_\lambda^d}\otimes I_{\mathrm{V}_\nu^d})U I_{\mathrm{V}_\mu^d})\\
    &=\frac{1}{\dim(\mathrm{V}_\mu^d)}\tr(I_{\mathrm{V}_\mu^d})\\
    &=1.
\end{align}

Suppose $d>l$ and assume that the claim holds in dimension $d-1$.
We decompose the LHS of Eq.~\eqref{eq:BCI-general} using the branching rule in Eq.~\eqref{eq:u_restrict}, apply the induction hypothesis, and reduce the desired inequality to a one-step inequality. 
To this end, we choose the Schur transform convention in which
$P_\lambda^{d,l}$ projects onto the subspace spanned by the GT basis elements
obtained by adjoining $d-l$ rows of $\lambda$ above a GT pattern of $\mathrm{V}_\lambda^l$, \textit{i.e.}, $P_\lambda^{d, l}=JJ^\dagger$ for an inclusion isometry $J:\mathrm{V}_\lambda^l\to \mathrm{V}_\lambda^d$ satisfying
\begin{align}
    J:\quad
    \left(
    \begin{array}{c}
        \gtbox{3.5cm}{{\lambda}}\\[0.35em]
        \gtbox{2.5cm}{{\lambda}^{l-1}}\\[0.35em]
        \vdots\\[0.35em]
        \gtbox{0.5cm}{{\lambda}^1}
    \end{array}
    \right)
    \to
    \left(
    \begin{array}{c}
        \gtbox{3.5cm}{{\lambda}}~
        \gtbox{3cm}{0^{d-l}}\\[0.35em]
        \gtbox{3.5cm}{{\lambda}}~
        \gtbox{2.0cm}{0^{d-l-1}}\\[0.35em]
        \qquad\ddots\qquad\qquad\qquad\qquad\vdots\\[0.35em]
        \gtbox{3.5cm}{{\lambda}}\\[0.35em]
        \gtbox{2.5cm}{{\lambda}^{l-1}}\\[0.35em]
        \vdots\\[0.35em]
        \gtbox{0.5cm}{{\lambda}^1}
    \end{array}
    \right),
\end{align}
where each row interlaces with the row immediately above it, and $0^k$ denotes the length-$k$ zero vector.
Then, for each $\gamma \prec \lambda$, the $(d-1)\to d$ inclusion isometry $J_{\gamma,\lambda}:\mathrm{V}_\gamma^{d-1}\to\mathrm{V}_\lambda^d$ acts as follows:
\begin{align}
    J_{\gamma,\lambda}:\quad
    \left(
    \begin{array}{c}
        \gtbox{3.5cm}{{\gamma}}\\[0.35em]
        \gtbox{2.5cm}{{\gamma}^{d-2}}\\[0.35em]
        \vdots\\[0.35em]
        \gtbox{0.5cm}{{\gamma}^1}
    \end{array}
    \right)
    \to
    \left(
    \begin{array}{c}
        \gtbox{4.5cm}{{\lambda}}\\[0.35em]
        \gtbox{3.5cm}{{\gamma}}\\[0.35em]
        \gtbox{2.5cm}{{\gamma}^{d-2}}\\[0.35em]
        \vdots\\[0.35em]
        \gtbox{0.5cm}{{\gamma}^1}
    \end{array}
    \right).
\end{align}

Using this convention, we decompose the LHS of Eq.~\eqref{eq:BCI-general}.
Since $d>l$, we have $\len(\lambda),\len(\mu)\leq l\le d-1$, allowing us to rewrite the projectors as follows:
\begin{align}
    P_{\lambda}^{d,l}=
    J_{\lambda} P_\lambda^{d-1, l} J_{\lambda}^\dagger,
    \qquad P_{\mu}^{d,l}=
    J_{\mu} P_\mu^{d-1, l} J_{\mu}^\dagger,
\end{align}
where we use the shorthand
$J_\lambda\coleq J_{\lambda,\lambda}$ and
$J_\mu\coleq J_{\mu,\mu}$ for simplicity.
The identity operator $I_{\mathrm{V}_\nu^d}$ similarly decomposes as
\begin{align}
    \label{eq:nu_id}
    I_{\mathrm{V}_\nu^d}
    = \sum_{\gamma\prec\nu} J_{\gamma,\nu} J_{\gamma,\nu}^\dagger.
\end{align}
For each $\gamma \prec \nu$, define the composite map $K_\gamma:\mathrm{V}_\mu^{d-1}\to\mathrm{V}_\lambda^{d-1}\otimes \mathrm{V}_\gamma^{d-1}$ by
\begin{align}
    K_\gamma\coleq (J_{\lambda}^\dagger\otimes J_{\gamma,\nu}^\dagger)U J_{\mu}.
    \label{eq:k_define}
\end{align}
Then, the trace term on the LHS of the lemma decomposes into a sum over $\gamma\prec\nu$:
\begin{align}
    \tr (U^\dagger(P_\lambda^{d, l}\otimes I_{\mathrm{V}_\nu^d})U P_\mu^{d, l})
    &=\tr \left(U^\dagger\left(J_{\lambda} P_\lambda^{d-1, l} J_{\lambda}^\dagger\otimes \sum_{\gamma\prec\nu} J_{\gamma,\nu} J_{\gamma,\nu}^\dagger\right)U 
    J_{\mu} P_\mu^{d-1, l} J_{\mu}^\dagger\right)\\
    &=\sum_{\gamma\prec\nu} \tr (K_\gamma^\dagger(P_\lambda^{d-1, l}\otimes I_{\mathrm{V}_\gamma^{d-1}})K_\gamma P_\mu^{d-1, l}).
    \label{eq:summand}
\end{align}

We next show that each nonzero $K_\gamma$ can be normalized to an isometric $\mathrm{U}(d-1)$-intertwiner, which allows us to apply the induction hypothesis.
From the $\mathrm{U}(d)\downarrow \mathrm{U}(d-1)\times\mathrm{U}(1)$ branching rule in Eq.~\eqref{eq:u_restrict}, the isometries $J_{\lambda}$, $J_{\gamma,\nu}$, and $J_{\mu}$ are $\mathrm{U}(d-1)$-intertwiners when the representations are restricted to $\mathrm{U}(d-1)$.
Moreover, since $U$ is a $\mathrm{U}(d)$-intertwiner, it is also a $\mathrm{U}(d-1)$-intertwiner under the same restriction.
Consequently, $K_\gamma$ is a $\mathrm{U}(d-1)$-intertwiner.
Schur's lemma then gives $K_\gamma^\dagger K_\gamma=c_\gamma I_{\mathrm{V}_{\mu}^{d-1}}$ for some $c_\gamma \ge 0$. 
Here, we can assume $c_\gamma>0$ since otherwise, the corresponding summand in Eq.~\eqref{eq:summand} vanishes. 
For $c_\gamma >0$, the map $c_\gamma^{-1/2}K_\gamma$ is an isometric $\mathrm{U}(d-1)$-intertwiner.
Together, we obtain the following upper bound on the LHS of Eq.~\eqref{eq:BCI-general}:

\begin{align}
    &\frac{\dim(\mathrm{V}_\lambda^d)}{\dim(\mathrm{V}_\lambda^l)\dim(\mathrm{V}_\mu^d)} \tr \left(U^\dagger(P_\lambda^{d, l}\otimes I_{\mathrm{V}_\nu^d})U P_\mu^{d, l}\right)\notag\\
    &\quad=\frac{\dim(\mathrm{V}_\lambda^d)}{\dim(\mathrm{V}_\lambda^l)\dim(\mathrm{V}_\mu^d)}\sum_{\gamma\prec\nu} c_\gamma \tr \left((c_\gamma^{-1/2}K_\gamma)^\dagger(P_\lambda^{d-1, l}\otimes I_{\mathrm{V}_\gamma^{d-1}})(c_\gamma^{-1/2}K_\gamma) P_\mu^{d-1, l}\right)\\
    &\quad\leq\frac{\dim(\mathrm{V}_\lambda^d)}{\dim(\mathrm{V}_\lambda^l)\dim(\mathrm{V}_\mu^d)}
    \frac{\dim(\mathrm{V}_\lambda^{l})\dim(\mathrm{V}_\mu^{d-1})}{\dim(\mathrm{V}_\lambda^{d-1})} \prod_{i=l+1}^{d-1}\frac{\lambda_1+i+M-2}{\lambda_1+i+2M-2} \sum_{\gamma\prec \nu}  c_\gamma\\
    &\quad=\frac{\dim(\mathrm{V}_\lambda^d ) \dim(\mathrm{V}_\mu^{d-1})}{\dim(\mathrm{V}_\mu^d)\dim(\mathrm{V}_\lambda^{d-1})}\prod_{i=l+1}^{d-1}\frac{\lambda_1+i+M-2}{\lambda_1+i+2M-2}\sum_{\gamma\prec\nu}\frac{\tr(K^\dagger_\gamma K_\gamma)}{\dim(\mathrm{V}_{\mu}^{d-1})}\\
    &\quad=\frac{\dim(\mathrm{V}_\lambda^d )}{\dim(\mathrm{V}_\mu^d)\dim(\mathrm{V}_\lambda^{d-1})}\prod_{i=l+1}^{d-1}\frac{\lambda_1+i+M-2}{\lambda_1+i+2M-2}\sum_{\gamma\prec \nu}\tr(U^\dagger(J_{\lambda}J_{\lambda}^\dagger\otimes J_{\gamma,\nu}J_{\gamma,\nu}^\dagger)UJ_{\mu} J_{\mu}^\dagger)\\
    &\quad=\frac{\dim(\mathrm{V}_\lambda^d )}{\dim(\mathrm{V}_\mu^d)\dim(\mathrm{V}_\lambda^{d-1})}\prod_{i=l+1}^{d-1}\frac{\lambda_1+i+M-2}{\lambda_1+i+2M-2}\tr(U^\dagger(P_{\lambda}^{d,d-1}\otimes I_{\mathrm{V}_\nu^d})U P_{\mu}^{d,d-1}),
\end{align}
where the inequality follows from the induction hypothesis with dimensions $d-1$ and $l$.

We now apply the following lemma:
\begin{lemma}[One-step inequality]\label{lem:bci_one_step}
    For partitions $\lambda\vdash_{d-1} N$ and $\mu\vdash_{d-1} N+M$, a length-$d$ staircase $\nu$, and an isometric $\mathrm{U}(d)$-intertwiner $U:\mathrm{V}_\mu^d\to \mathrm{V}_\lambda^d\otimes \mathrm{V}_\nu^d$, the following inequality holds:
    \begin{align}
    \frac{\dim(\mathrm{V}_\lambda^d)}{\dim(\mathrm{V}_\mu^d)\dim(\mathrm{V}_\lambda^{d-1})}
    \tr (
        U^\dagger(P_\lambda^{d,d-1}\otimes I_{\mathrm{V}_\nu^d})
        U P_\mu^{d,d-1}
    )
    \leq 
    \frac{\lambda_1+d+M-2}{\lambda_1+d+2M-2}.
    \label{eq:bci_one_step}
    \end{align}
\end{lemma}
\begin{proof}
    See Sec.~\ref{sec:l8proof}.
\end{proof}
\noindent
Applying Lemma~\ref{lem:bci_one_step} to the preceding bound, we obtain
\begin{align}
    \frac{\dim(\mathrm{V}_\lambda^d)}{\dim(\mathrm{V}_\mu^d)\dim(\mathrm{V}_\lambda^l)} \tr (U^\dagger(P_\lambda^{d, l}\otimes I_{\mathrm{V}_\nu^d})U P_\mu^{d, l})
    &\leq \left(\prod_{i=l+1}^{d-1}\frac{\lambda_1+i+M-2}{\lambda_1+i+2M-2}\right)\frac{\lambda_1+d+M-2}{\lambda_1+d+2M-2}\\
    &=\prod_{i=l+1}^{d}\frac{\lambda_1+i+M-2}{\lambda_1+i+2M-2},
\end{align}
which completes the proof.

\subsection{Proof of Lemma~\ref{lem:bci_one_step}}
\label{sec:l8proof}

    The proof proceeds in two parts: first, we rewrite the inequality as a simpler matrix-norm inequality, and second, we prove the new inequality using Lie algebra representation theory.

    We first rewrite the trace term on the LHS of Eq.~\eqref{eq:bci_one_step} in a more convenient form.
    The key observation is that the intertwiner $U$ preserves the $\mathrm{U}(1)$-weight associated with the restriction $\mathrm{U}(d)\downarrow \mathrm{U}(d-1)\times\mathrm{U}(1)$.
    To see this, consider a basis vector
    \begin{align}
        \ket{v}=
        \left(
        \begin{array}{c}
            \gtbox{2.5cm}{{\mu}}\\[0.35em]
            \gtbox{1.5cm}{{\mu^{d-1}}}\\[0.35em]
            \vdots
        \end{array}
        \right)
    \end{align}
    of $\mathrm{V}_\mu^d$.
    By the branching rule in Eq.~\eqref{eq:u_restrict}, under the action of $I_{d-1}\oplus z$ with $z\in\mathrm{U}(1)$, the vector $\ket{v}$ has $\mathrm{U}(1)$-weight $|\mu|-|\mu^{d-1}|$:
    \begin{align}
        \varphi_{\mu}(I_{d-1}\oplus z)\ket{v}
        =
        z^{|\mu|-|\mu^{d-1}|}\ket{v}.
    \end{align}
    Using the intertwining property of $U$, we then have
    \begin{align}
        z^{|\mu|-|\mu^{d-1}|}U\ket{v}
        &=
        U\varphi_{\mu}(I_{d-1}\oplus z)\ket{v}\\
        &=
        (
            \varphi_{\lambda}(I_{d-1}\oplus z)
            \otimes
            \varphi_{\nu}(I_{d-1}\oplus z)
        )U\ket{v}
    \end{align}
    for every $z\in\mathrm{U}(1)$.
    Hence, $U\ket{v}$ is supported only on tensor-product basis vectors whose total $\mathrm{U}(1)$-weight equals that of $\ket{v}$, \textit{i.e.},
    \begin{align}
        \label{eq:u_preserve}
        U\ket{v}
        \in\mathrm{span}\left(\left\{
        \left(
        \begin{array}{c}
            \gtbox{2.5cm}{{\lambda}}\\[0.35em]
            \gtbox{1.5cm}{{\lambda^{d-1}}}\\[0.35em]
            \vdots
        \end{array}
        \right)
        \otimes
        \left(
        \begin{array}{c}
            \gtbox{2.5cm}{{\nu}}\\[0.35em]
            \gtbox{1.5cm}{{\nu^{d-1}}}\\[0.35em]
            \vdots
        \end{array}
        \right)
        :
        |\lambda|-|\lambda^{d-1}|
        +
        |\nu|-|\nu^{d-1}|
        =
        |\mu|-|\mu^{d-1}|
        \right\}\right).
    \end{align}
    In other words, the intertwiner $U$ preserves the total $\mathrm{U}(1)$-weight under the restricted action.

    This $\mathrm{U}(1)$-weight preservation allows us to characterize the trace term of interest in a more convenient form.
    Let $P^{d,d-1}_{\nu,\gamma}\coleq J_{\gamma,\nu}J_{\gamma,\nu}^\dagger$
    and denote the projector onto the weight-$q$ subspace by
    \begin{align}
        Q_{\nu,q}
        \coleq
        \sum_{\substack{\gamma\prec\nu\\|\nu|-|\gamma|=q}}
        P^{d,d-1}_{\nu,\gamma}.
    \end{align}
    Since $P_\lambda^{d,d-1}$ and $P_\mu^{d,d-1}$ project onto the corresponding weight-zero subspaces, only the weight-zero component of $I_{\mathrm{V}_\nu^d}$ contributes, yielding
    \begin{align}
        (P_\lambda^{d,d-1}\otimes I_{\mathrm{V}_\nu^d})
            U P_\mu^{d,d-1}
        =
        (P_\lambda^{d,d-1}\otimes Q_{\nu,0})
        U P_\mu^{d,d-1}.
    \end{align}
    Here, from the assumption, $\lambda$ and $\mu$ are both partitions of length at most $d-1$, and hence the images of $P_\lambda^{d,d-1}$ and $P_\mu^{d,d-1}$ are precisely the corresponding weight-zero subspaces.
    This yields $P_\lambda^{d,d-1}=Q_{\lambda,0}$ and $P_\mu^{d,d-1}=Q_{\mu,0}$, which implies the following form of the trace term:
    \begin{align}
        \tr (
            U^\dagger(P_\lambda^{d,d-1}\otimes I_{\mathrm{V}_\nu^d})
            U P_\mu^{d,d-1}
        )
        =
        \tr (
            U^\dagger(Q_{\lambda,0}\otimes Q_{\nu,0})
            U Q_{\mu,0}
        ).
    \end{align}
    Equivalently, denoting
    \begin{align}
        B_0\coleq (Q_{\lambda,0}\otimes Q_{\nu,0})
            U Q_{\mu,0},
    \end{align}
    we obtain
    \begin{align}
            \tr (
                U^\dagger(P_\lambda^{d,d-1}\otimes I_{\mathrm{V}_\nu^d})
                U P_\mu^{d,d-1}
            )
            =\|B_0\|_{\mathrm{HS}}^2,
    \end{align}
    where $\|\cdot\|_{\mathrm{HS}}$ denotes the Hilbert-Schmidt norm, given by $\|A\|_{\mathrm{HS}}\coleq(\tr(A^\dagger A))^{1/2}$.

    The same weight-based decomposition also allows us to rewrite the prefactor on the LHS of Eq.~\eqref{eq:bci_one_step}.
    For each weight $q$, define
    \begin{align}
        B_q\coleq (Q_{\lambda,0}\otimes Q_{\nu,q})UQ_{\mu,q}.
    \end{align}
    The inverse of the prefactor can then be expressed as the sum of the squared Hilbert-Schmidt norms of the operators $B_q$:
    \begin{align}
         \sum_{q\in\mathbb Z}\| B_q\|_{\mathrm{HS}}^2
         &= \sum_{q \in \mathbb{Z}} \tr \left( U^\dagger (Q_{\lambda,0} \otimes Q_{\nu,q} )U Q_{\mu, q}\right)\\
         &=\sum_{q \in \mathbb{Z}} \tr \left( U^\dagger (P_\lambda^{d,d-1} \otimes I_{\mathrm{V}_\nu^d} )U Q_{\mu, q}\right)\\
         &=\tr \left( U^\dagger (P_\lambda^{d,d-1} \otimes I_{\mathrm{V}_\nu^d} )U I_{\mathrm{V}_\mu^d}\right)\\
         &=\tr\left((P_\lambda^{d,d-1}\otimes I_{\mathrm{V}_\nu^d})UU^\dagger\right)\\
         &= \int_{\mathrm{U}(d)} \tr \left( (\varphi_\lambda(W) P_{\lambda}^{d, d-1} \varphi_\lambda(W)^\dagger \otimes I_{\mathrm{V}_\nu^d}) UU^\dagger \right)  dW\\
         &=\frac{\dim(\mathrm{V}_\lambda^{d-1})}{\dim(\mathrm{V}_\lambda^d)}\tr(UU^\dagger)\\
         &=\frac{\dim(\mathrm{V}_\mu^d)\dim(\mathrm{V}_\lambda^{d-1})}
         {\dim(\mathrm{V}_\lambda^d)},
    \end{align}
    where the third line follows from $\sum_{q\in\mathbb{Z}}Q_{\mu,q}=I_{\mathrm{V}_\mu^d}$, the fifth uses the $\mathrm{U}(d)$-invariance of $UU^\dagger$, and the sixth follows from Schur's lemma.

    Using the preceding identities, Eq.~\eqref{eq:bci_one_step} can be rewritten as
    \begin{align}
        \| B_0\|_{\mathrm{HS}}^2
        \leq\frac{\lambda_1+d+M-2}{\lambda_1+d+2M-2}
        \sum_{q\in\mathbb{Z}}\| B_q\|_{\mathrm{HS}}^2
    \end{align}
    or equivalently,
    \begin{align}
        \| B_0\|_{\mathrm{HS}}^2
        \leq\frac{\lambda_1+d+M-2}{M}
        \sum_{q\neq 0}\| B_q\|_{\mathrm{HS}}^2.
    \end{align}

    We now prove this inequality by establishing the stronger inequality:
    \begin{align}
        \| B_0\|_{\mathrm{HS}}^2
        \leq
        \frac{\lambda_1+d+M-2}{M}\| B_1\|_{\mathrm{HS}}^2.
        \label{eq:bci_b0_b1}
    \end{align} 
    To this end, we relate the operators $B_0$ and $B_1$, which lie in different $\mathrm{U}(1)$-weight subspaces, using infinitesimal Lie algebra actions.
    These actions in particular provide operators that increase the $\mathrm{U}(1)$ weight.
    To make this explicit, for a standard matrix unit $\ketbra{i}{j}\in\mathfrak{gl}_d(\mathbb C)$, let $E_\kappa^{i,j}$ denote the corresponding infinitesimal action on $\mathrm{V}_\kappa^d$.
    Consider a GT basis vector $\ket{v}$ with its $k$-th row given by $\kappa^k$ and let $\kappa^0=0$.
    Then, the diagonal operator $E_\kappa^{k,k}$ acts on $\ket{v}$ by
    \begin{align}
        \label{eq:ekk}
        E_\kappa^{k,k}\ket{v}
        =(|\kappa^k|-|\kappa^{k-1}|)\ket{v},
    \end{align}
    as follows from~\cite[Theorem 2.3]{molev2006gelfand}.
    In particular, for $k=d$, $E_\kappa^{d,d}$ acts on the weight-$q$ subspace as multiplication by $q$:
    \begin{align}
        E_\kappa^{d,d}Q_{\kappa,q}=qQ_{\kappa,q}.
    \end{align}
    Consequently, from the standard commutation relation $[E_\kappa^{d,d}, E_\kappa^{d,i}]=E_\kappa^{d,i}$ for $i< d$, it follows that
    \begin{align}
        E_\kappa^{d,d} E_\kappa^{d,i}Q_{\kappa,q}
        &=(E_\kappa^{d,i} E_\kappa^{d,d}+E_\kappa^{d,i})Q_{\kappa,q}\\
        &=(q+1)E_\kappa^{d,i}Q_{\kappa,q},\label{eq:weight_plus}
    \end{align}
    showing that $E_\kappa^{d,i}$ maps the weight-$q$ subspace into the weight-$(q+1)$ subspace and thus increases the $\mathrm{U}(1)$ weight by one.
    Using these properties, we proceed with the rest of the proof in three steps.
    \begin{lemma}
        \label{lem:step1}
        \begin{align}
            \|B_0\|^2_{\mathrm{HS}}\leq \frac{\lambda_1+d+M-2}{M(\mu_1+d-2)}\sum_{i=1}^{d-1}\|(I_{\mathrm{V}_\lambda^d}\otimes E_\nu^{d,i})B_0\|_{\mathrm{HS}}^2.
        \end{align}
    \end{lemma}
    \begin{proof}
        See Sec.~\ref{sec:l9proof}.
    \end{proof}

    \begin{lemma}
        \label{lem:step2}
        \begin{align}
        \sum_{i=1}^{d-1}
        \|(I_{\mathrm V_\lambda^d}\otimes E_\nu^{d,i})B_0\|_{\mathrm{HS}}^2
        =
        \sum_{i=1}^{d-1}
        \|B_1E_\mu^{d,i}Q_{\mu,0}\|_{\mathrm{HS}}^2.
        \end{align}
    \end{lemma}
    \begin{proof}
        See Sec.~\ref{sec:l10proof}.
    \end{proof}

    \begin{lemma}
        \label{lem:step3}
        \begin{align}
            \sum_{i=1}^{d-1}
            \|B_1E_\mu^{d,i}Q_{\mu,0}\|_{\mathrm{HS}}^2
            \leq
            (\mu_1+d-2)\|B_1\|_{\mathrm{HS}}^2.
        \end{align}
    \end{lemma}
    \begin{proof}
        See Sec.~\ref{sec:l11proof}.
    \end{proof}
    Combining the three lemmas yields Eq.~\eqref{eq:bci_b0_b1}, thereby completing the proof.

    \subsection{Proof of Lemma~\ref{lem:step1}}\label{sec:l9proof}
    We decompose the operator $B_0$ according to the $\mathrm U(d-1)$ branching rule and prove that
    \begin{align}
        \label{eq:step1}
        \|B_0\|^2_{\mathrm{HS}}\leq \frac{\lambda_1+d+M-2}{M(\mu_1+d-2)}\sum_{i=1}^{d-1}\|(I_{\mathrm{V}_\lambda^d}\otimes E_\nu^{d,i})B_0\|_{\mathrm{HS}}^2.
    \end{align}
    Decompose $B_0$ as
    \begin{align}
        B_0
        =\sum_{\substack{\gamma\prec\nu\\|\nu|-|\gamma|=0}}
        B_{0,\gamma},\qquad B_{0,\gamma}\coleq (I_{\mathrm{V}_\lambda^d}\otimes P_{\nu,\gamma}^{d,d-1})B_0.
    \end{align}
    Then, by the orthogonality of the projectors,
    \begin{align}
        \label{eq:step1-1}
        \|B_0\|^2_{\mathrm{HS}}
        =\sum_{\substack{\gamma\prec\nu\\|\nu|-|\gamma|=0}}\|B_{0,\gamma}\|^2_{\mathrm{HS}}.
    \end{align}

    We first establish the following auxiliary inequality:
    \begin{align}
        \label{eq:step1-2}
        \sum_{\substack{\gamma\prec\nu\\|\nu|-|\gamma|=0}}\|B_{0,\gamma}\|^2_{\mathrm{HS}}
        \leq \frac{\lambda_1+d+M-2}{M(\mu_1+d-2)}\sum_{\substack{\gamma\prec\nu\\|\nu|-|\gamma|=0}}\max\{\gamma_1,M\}\|B_{0,\gamma}\|^2_{\mathrm{HS}}.
    \end{align}
    Since each $\|B_{0,\gamma}\|_{\mathrm{HS}}^2$ is nonnegative, it suffices to show that, for every nonzero $B_{0,\gamma}$, the corresponding coefficient on the RHS is at least $1$, \textit{i.e.},
    \begin{align}
        \label{eq:one-coeff}
        \frac{\lambda_1+d+M-2}{M(\mu_1+d-2)}\max\{\gamma_1,M\}\geq 1.
    \end{align}
    This follows from combining several inequalities. 
    For a nonzero $B_{0,\gamma}$, the identity $B_{0,\gamma}=(J_\lambda\otimes J_{\gamma,\nu})K_\gamma J_\mu^\dagger$ implies that $K_\gamma$ is nonzero.
    Since $K_\gamma$ is a $\mathrm U(d-1)$-intertwiner with irreducible domain $\mathrm V_\mu^{d-1}$, it is injective.
    Thus, $\mathrm V_\mu^{d-1}$ occurs in $\mathrm V_\lambda^{d-1}\otimes\mathrm V_\gamma^{d-1}$.
    The CG rule~\cite[Appendix~A.8]{fulton2013representation} then gives
    \begin{align}
        \mu_1\leq\lambda_1+\gamma_1,
    \end{align}
    which allows us to show that
    \begin{align}
        \frac{\lambda_1+d+M-2}{M(\lambda_1+\gamma_1+d-2)}\max\{\gamma_1,M\}\geq 1
    \end{align}
    instead.
    If $\gamma_1\leq M$, the inequality follows from
    \begin{align}
        \frac{\lambda_1+d+M-2}{M(\lambda_1+\gamma_1+d-2)}\max\{\gamma_1,M\}
        \geq \frac{1}{M}\max\{\gamma_1,M\} = 1.
    \end{align}
    If $\gamma_1>M$,
    \begin{align}
        (\lambda_1+d+M-2)\max\{\gamma_1,M\}-M(\lambda_1+\gamma_1+d-2)
        &=(\gamma_1-M)(\lambda_1+d-2)\geq 0
    \end{align}
    follows from $\lambda_1+d-2\geq 0$, which implies the inequality.
    Therefore, we obtain Eq.~\eqref{eq:one-coeff} and consequently Eq.~\eqref{eq:step1-2}.

    We now use Lie algebra representation theory to prove the following inequality:
    \begin{align}
        \label{eq:step1-3}
        \max\{\gamma_1,M\}\|B_{0, \gamma}\|^2_{\mathrm{HS}}\leq\sum_{i=1}^{d-1}\|(I_{\mathrm{V}_\lambda^d}\otimes E_\nu^{d,i})B_{0, \gamma}\|_{\mathrm{HS}}^2
    \end{align}
    for $\gamma\prec \nu$ with $|\nu|-|\gamma|=0$.
    Rewriting this in terms of trace gives the following equivalent inequality:
    \begin{align}
        \max\{\gamma_1,M\}\tr(B_{0,\gamma}^\dagger B_{0,\gamma})
        &\leq \tr\left(B_{0,\gamma}^\dagger\left(I_{\mathrm{V}_\lambda^d}\otimes \sum_{i=1}^{d-1}E_\nu^{i,d}E_\nu^{d,i}\right)B_{0,\gamma}\right)\\
        &=\tr\left(B_{0,\gamma}^\dagger\left(I_{\mathrm{V}_\lambda^d}\otimes P_{\nu,\gamma}^{d,d-1}\left(\sum_{i=1}^{d-1}E_\nu^{i,d}E_\nu^{d,i}\right)P_{\nu,\gamma}^{d,d-1}\right)B_{0,\gamma}\right),
    \end{align}
    where we used $(E_\nu^{d,i})^\dagger=E_\nu^{i,d}$ and the definition of $B_{0,\gamma}$.
    Hence, it suffices to show that
    \begin{align}
        \label{eq:gamma1-bound}
        \max\{\gamma_1,M\} P_{\nu,\gamma}^{d,d-1} \leq P_{\nu,\gamma}^{d,d-1}\left(\sum_{i=1}^{d-1}E_\nu^{i,d}E_\nu^{d,i}\right)P_{\nu,\gamma}^{d,d-1}.
    \end{align}

    We prove this operator inequality by showing that the RHS acts as a scalar on the image of $P_{\nu,\gamma}^{d,d-1}$.
    By the branching rule, the image of $P_{\nu,\gamma}^{d,d-1}$ is a $\mathrm U(d-1)$-invariant subspace isomorphic to $\mathrm V_\gamma^{d-1}\otimes\mathbb C_{|\nu|-|\gamma|}$.
    Thus, the projector $P_{\nu,\gamma}^{d,d-1}$ commutes with the $\mathrm U(d-1)$ action.
    The summation term also commutes with the $\mathrm{U}(d-1)$ action as follows:
    \begin{align}
        \left[E_\nu^{j,k}, \sum_{i=1}^{d-1}E_\nu^{i,d}E_\nu^{d,i}\right]
        &=\sum_{i=1}^{d-1}[E_\nu^{j,k}, E_\nu^{i,d}E_\nu^{d,i}]\\
        &=\sum_{i=1}^{d-1}([E_\nu^{j,k}, E_\nu^{i,d}]E_\nu^{d,i}+E_\nu^{i,d}[E_\nu^{j,k}, E_\nu^{d,i}])\\
        &=\sum_{i=1}^{d-1}(\delta_{k,i}E_{\nu}^{j,d}E_{\nu}^{d,i}-\delta_{j, i}E_\nu^{i, d}E_\nu^{d,k})\\
        &=0\label{eq:chobap}
    \end{align}
    for $j,k<d$, where $\delta_{\cdot,\cdot}$ is the Kronecker delta, and the third line follows from the standard commutation relation $[E_\nu^{a,b}, E_\nu^{c,d}]=\delta_{b,c}E_\nu^{a,d}-\delta_{d,a}E_\nu^{c,b}$.
    Therefore, by Schur's lemma,
    \begin{align}
        P_{\nu,\gamma}^{d,d-1}
        \left(\sum_{i=1}^{d-1}E_\nu^{i,d}E_\nu^{d,i}\right)
        P_{\nu,\gamma}^{d,d-1}
        =
        c\,P_{\nu,\gamma}^{d,d-1}
    \end{align}
    for some scalar $c$.

    It therefore remains to show that $c\geq\gamma_1$ and $c\geq M$.
    To first show $c\geq\gamma_1$, choose the following unit GT basis vector
    \begin{align}
        \ket{v}=
        \left(
        \begin{array}{c}
            \gtbox{3.5cm}{\nu}\\[0.35em]
            \gtbox{2.5cm}{\gamma}\\[0.35em]
            \vdots\\[0.35em]
            \gtbox{0.5cm}{(\gamma_1)}
        \end{array}
        \right).
    \end{align}
    Then, we obtain
    \begin{align}
        c
        &=\bra{v}\sum_{i=1}^{d-1}E_\nu^{i,d}E_\nu^{d,i}\ket{v}\\
        &=\sum_{i=1}^{d-1}\|E_\nu^{d,i}\ket{v}\|^2\\
        &\geq\|E_\nu^{d,1}\ket{v}\|^2\\
        &=\bra{v}E_\nu^{1,d}E_\nu^{d,1}\ket{v}\\
        &=\bra{v}(E_\nu^{d,1}E_\nu^{1,d}+E_\nu^{1,1}-E_\nu^{d,d})\ket{v}\\
        &=\|E_\nu^{1,d}\ket{v}\|^2+\gamma_1-(|\nu|-|\gamma|)\\
        &\geq \gamma_1,
    \end{align}
    where we used $(E_\nu^{d,i})^\dagger=E_\nu^{i,d}$, the commutation relation $[E_\nu^{1,d},E_\nu^{d,1}]=E_\nu^{1,1}-E_\nu^{d,d}$, Eq.~\eqref{eq:ekk}, and $|\nu|-|\gamma|=0$.

    We then show that $c\geq M$.
    From the commutation relation, we obtain
    \begin{align}
        c
        &=\bra{v}\sum_{i=1}^{d-1}E_\nu^{i,d}E_\nu^{d,i}\ket{v}\\
        &=\sum_{i=1}^{d-1}\bra{v}(E_\nu^{d,i}E_\nu^{i,d}+E_\nu^{i,i}-E_\nu^{d,d})\ket{v}\\
        &=\sum_{i=1}^{d-1}\|E_\nu^{i,d}\ket{v}\|^2+|\gamma|-(d-1)(|\nu|-|\gamma|)\\
        &=\sum_{i=1}^{d-1}\|E_\nu^{i,d}\ket{v}\|^2+M\\
        &\geq M,
    \end{align}
    where we also used $(E_\nu^{i,d})^\dagger=E_\nu^{d,i}$, $|\nu|-|\gamma|=0$, Eq.~\eqref{eq:ekk}, and the size matching condition $|\nu|=|\mu|-|\lambda|=M$ in Lemma~\ref{lem:sizematching}.
    This proves Eq.~\eqref{eq:gamma1-bound}, and consequently Eq.~\eqref{eq:step1-3}.

    Combining Eqs.~\eqref{eq:step1-1}, \eqref{eq:step1-2},
    and \eqref{eq:step1-3} yields
    \begin{align}
        \|B_0\|^2_{\mathrm{HS}}\leq \frac{\lambda_1+d+M-2}{M(\mu_1+d-2)} \sum_{\substack{\gamma\prec\nu\\|\nu|-|\gamma|=0}}\sum_{i=1}^{d-1}\|(I_{\mathrm{V}_\lambda^d}\otimes E_\nu^{d,i})B_{0, \gamma}\|_{\mathrm{HS}}^2.
    \end{align}
    We complete the proof by showing that
    \begin{align}
        \sum_{\substack{\gamma\prec\nu\\|\nu|-|\gamma|=0}}\sum_{i=1}^{d-1}\|(I_{\mathrm{V}_\lambda^d}\otimes E_\nu^{d,i})B_{0, \gamma}\|_{\mathrm{HS}}^2=\sum_{i=1}^{d-1}
        \|(I_{\mathrm V_\lambda^d}\otimes E_\nu^{d,i})B_0\|_{\mathrm{HS}}^2.
        \label{eq:step1-4}
    \end{align}
    As shown in Eq.~\eqref{eq:chobap}, the operator $\sum_{i=1}^{d-1}E_\nu^{i,d}E_\nu^{d,i}$ commutes with the $\mathrm U(d-1)$ action.
    Hence, for distinct $\gamma,\gamma'\prec\nu$,
    \begin{align}
        P_{\nu,\gamma}^{d,d-1}
        \left(\sum_{i=1}^{d-1}E_\nu^{i,d}E_\nu^{d,i}\right)
        P_{\nu,\gamma'}^{d,d-1}
        =0.
    \end{align}
    Indeed, the LHS is a $\mathrm U(d-1)$-intertwiner
    from $\mathrm V_{\gamma'}^{d-1}$ to $\mathrm V_\gamma^{d-1}$ up to isomorphism, and therefore
    vanishes by Schur's lemma for $\gamma\neq\gamma'$.
    From the definition of $B_{0,\gamma}$,
    we therefore have
    \begin{align}
        &\sum_{i=1}^{d-1}
        \tr\left(
            B_{0,\gamma}^\dagger
            (I_{\mathrm V_\lambda^d}\otimes E_\nu^{i,d}E_\nu^{d,i})
            B_{0,\gamma'}
        \right)
        =0.
    \end{align}
    This implies Eq.~\eqref{eq:step1-4}, completing the proof.

    \subsection{Proof of Lemma~\ref{lem:step2}}\label{sec:l10proof}
    We show that 
    \begin{align}
        \sum_{i=1}^{d-1}
        \|(I_{\mathrm V_\lambda^d}\otimes E_\nu^{d,i})B_0\|_{\mathrm{HS}}^2
        =
        \sum_{i=1}^{d-1}
        \|B_1E_\mu^{d,i}Q_{\mu,0}\|_{\mathrm{HS}}^2
    \end{align}
    by combining relations for the infinitesimal Lie algebra actions.
    First, we have
    \begin{align}
        \label{eq:step2-1}
        E_\nu^{d,i}Q_{\nu,0}=Q_{\nu,1}E_\nu^{d,i}Q_{\nu,0},
    \end{align}
    for $i<d$, which follows from the fact that $E_\nu^{d,i}$ increases the $\mathrm{U}(1)$-weight by one, as shown in Eq.~\eqref{eq:weight_plus}.
    The same weight-raising property also gives
    \begin{align}
        \label{eq:step2-2}
        Q_{\lambda,0}E_\lambda^{d,i}=0,
    \end{align}
    because $\mathrm V_\lambda^d$ has no weight-$(-1)$ subspace.
    Finally, the intertwining property of $U$ gives
    \begin{align}
        \label{eq:step2-3}
        UE^{d,i}_{\mu} 
        =
        (E^{d,i}_{\lambda}\otimes I_{\mathrm{V}_\nu^d}+I_{\mathrm{V}_\lambda^d}\otimes E^{d,i}_{\nu})U,
    \end{align}
    for $i<d$, obtained by differentiating $U\varphi_\mu(g) =(\varphi_\lambda(g)\otimes\varphi_\nu(g))U$ and extending the resulting relation complex-linearly.
    Using these identities, we obtain
    \begin{align}
        (I_{\mathrm{V}_\lambda^d}\otimes E_\nu^{d,i})B_0
        &=(Q_{\lambda,0}\otimes E_\nu^{d,i}Q_{\nu,0})UQ_{\mu,0}\\
        &=(Q_{\lambda,0}\otimes Q_{\nu,1})(I_{\mathrm{V}_{\lambda}^d}\otimes E_\nu^{d,i})UQ_{\mu,0}\\
        &=(Q_{\lambda,0}\otimes Q_{\nu,1})(UE^{d,i}_{\mu}- (E^{d,i}_{\lambda}\otimes I_{\mathrm{V}_\nu^d})U)Q_{\mu,0}\\
        &=(Q_{\lambda,0}\otimes Q_{\nu,1})UE^{d,i}_{\mu}Q_{\mu,0}\\
        &=(Q_{\lambda,0}\otimes Q_{\nu,1})UQ_{\mu,1}E^{d,i}_{\mu}Q_{\mu,0}\\
        &=B_1E^{d,i}_{\mu}Q_{\mu,0},
    \end{align}
    where the first and last lines follow from the definition of $B_q$, the second and fifth follow from Eq.~\eqref{eq:step2-1} and its analogue for $\mu$, respectively, the third follows from Eq.~\eqref{eq:step2-3}, and the fourth follows from Eq.~\eqref{eq:step2-2}.

    \subsection{Proof of Lemma~\ref{lem:step3}}\label{sec:l11proof}
    We show that
    \begin{align}
            \sum_{i=1}^{d-1}
            \|B_1E_\mu^{d,i}Q_{\mu,0}\|_{\mathrm{HS}}^2
            \leq
            (\mu_1+d-2)\|B_1\|_{\mathrm{HS}}^2.
    \end{align}
        
    We first rewrite the LHS as
    \begin{align}
        \mathrm{LHS}
        &=\sum_{i=1}^{d-1}
        \tr(Q_{\mu,0}E_\mu^{i,d}B_1^\dagger B_1E_\mu^{d,i}Q_{\mu,0})\\
        &=\sum_{i=1}^{d-1}
        \tr(B_1^\dagger B_1E_\mu^{d,i}Q_{\mu,0}E_\mu^{i,d})\\
        &=\sum_{i=1}^{d-1}
        \tr(B_1^\dagger B_1E_\mu^{d,i}E_\mu^{i,d}Q_{\mu,1})\\
        &=\tr\left(
            B_1^\dagger B_1
            \sum_{i=1}^{d-1}E_\mu^{d,i}E_\mu^{i,d}
        \right),
    \end{align}
    where the third line follows from $Q_{\mu,0}E_\mu^{i,d}=E_\mu^{i,d}Q_{\mu,1}$, and the last line follows from $B_1=B_1Q_{\mu,1}$.
    Since $B_1^\dagger B_1$ is positive semidefinite and supported on the weight-one subspace $Q_{\mu,1}\mathrm V_\mu^d$, it suffices to prove the following operator inequality on that subspace:
    \begin{align}
        \left.\sum_{i=1}^{d-1}E_\mu^{d,i}E_\mu^{i,d}\right|_{Q_{\mu,1}\mathrm{V}_\mu^d}\leq (\mu_1+d-2)I_{Q_{\mu,1}\mathrm{V}_\mu^d}.
    \end{align}

    We prove this by reducing it to a scalar inequality using Schur's lemma.
    To this end, we first show that $\sum_{i=1}^{d-1}E_\mu^{d,i}E_\mu^{i,d}$ commutes with the $\mathrm U(d-1)$ action.
    As shown in Eq.~\eqref{eq:chobap}, $\sum_{i=1}^{d-1}E_\mu^{i,d}E_\mu^{d, i}$ commutes with the $\mathrm{U}(d-1)$ action.
    Moreover, from the commutation relation $[E_\mu^{i,d},E_\mu^{d,i}]=E_\mu^{i,i}-E_\mu^{d,d}$, we have
    \begin{align}
        \sum_{i=1}^{d-1}E_\mu^{d,i}E_\mu^{i,d}
        &=
        \sum_{i=1}^{d-1}E_\mu^{i,d}E_\mu^{d,i}
        -
        \sum_{i=1}^{d-1}
        (E_\mu^{i,i}-E_\mu^{d,d}).
    \end{align}
    Since both $\sum_{i=1}^{d-1}E_\mu^{i,i}$ and $E_\mu^{d,d}$ commute with the $\mathrm U(d-1)$ action, it follows that $\sum_{i=1}^{d-1}E_\mu^{d,i}E_\mu^{i,d}$ also commutes with the $\mathrm U(d-1)$ action.

    We then decompose the weight-one subspace $Q_{\mu,1}\mathrm V_\mu^d$. 
    By the branching rule in Eq.~\eqref{eq:u_restrict}, the weight-one subspace decomposes as
    \begin{align}
        Q_{\mu,1}\mathrm V_\mu^d
        \simeq
        \bigoplus_{\substack{\gamma\prec\mu\\|\mu|-|\gamma|=1}}
        \mathrm V_\gamma^{d-1}\otimes\mathbb C_1
    \end{align}
    as a $\mathrm U(d-1)\times\mathrm U(1)$ representation.
    The staircases $\gamma\prec\mu$ satisfying $|\mu|-|\gamma|=1$ are precisely
    $\gamma=\mu-e_s$, where $1\leq s\leq d-1$ satisfies
    $\mu_s>\mu_{s+1}$.
    Thus, as a $\mathrm U(d-1)$ representation,
    \begin{align}
        Q_{\mu,1}\mathrm V_\mu^d
        \simeq
        \bigoplus_s\mathrm V_{\mu-e_s}^{d-1},
    \end{align}
    where $e_s$ denotes the length-$(d-1)$ tuple whose $s$-th entry is one and
    whose remaining entries are zero, and the direct sum is over the admissible indices $s$ specified above.
    
    Hence, Schur's lemma gives, for each admissible $s$,
    \begin{align}
        \left.
        \sum_{i=1}^{d-1}E_\mu^{d,i}E_\mu^{i,d}
        \right|_{\mathrm V_{\mu-e_s}^{d-1}}
        =
        c_s I_{\mathrm V_{\mu-e_s}^{d-1}}
    \end{align}
    for some scalar $c_s$.
    Therefore, it suffices to show that $c_s\leq\mu_1+d-2$ for every admissible index $s$.

    A useful observation is that the operator $\sum_{i,j=1}^{d}E_\mu^{i,j}E_\mu^{j,i}$ is the quadratic Casimir operator, which is well studied in representation theory.
    We compute $c_s$ using the formula for the eigenvalue of the quadratic Casimir operator given in Sec.~2.4 of Ref.~\cite{molev2006gelfand}:
    \begin{align}
        \left.\sum_{i,j=1}^{d}E_\mu^{i,j}E_\mu^{j,i}\right|_{\mathrm{V}_\mu^d}=\left(\sum_{k=1}^d\mu_k(\mu_k+d-2k+1)\right)I_{\mathrm{V}_\mu^d}.
    \end{align}
    Applying the corresponding formulas for $\mathrm U(d)$ and $\mathrm U(d-1)$ to the subspace $\mathrm V_{\mu-e_s}^{d-1}$, we obtain
    \begin{align}
        \left.\sum_{i,j=1}^{d}E_\mu^{i,j}E_\mu^{j,i}\right|_{\mathrm{V}_{\mu-e_s}^{d-1}}
        &=\left(\sum_{k=1}^d\mu_k(\mu_k+d-2k+1)\right)I_{\mathrm{V}_{\mu-e_s}^{d-1}},\\
        \left.\sum_{i,j=1}^{d-1}E_\mu^{i,j}E_\mu^{j,i}\right|_{\mathrm{V}_{\mu-e_s}^{d-1}}
        &=\left(\sum_{k=1}^{d-1}(\mu_k-\delta_{k,s})(\mu_k-\delta_{k,s}+d-2k)\right)I_{\mathrm{V}_{\mu-e_s}^{d-1}}.
    \end{align}
    It now remains to express $c_s$ in terms of these two Casimir operators.
    We have
    \begin{align}
        \sum_{i,j=1}^{d}E_\mu^{i,j}E_\mu^{j,i}-\sum_{i,j=1}^{d-1}E_\mu^{i,j}E_\mu^{j,i}
        &=(E_\mu^{d,d})^2+\sum_{i=1}^{d-1}E_\mu^{d,i}E_\mu^{i,d}+\sum_{i=1}^{d-1}E_\mu^{i,d}E_\mu^{d,i}\\
        &=(E_\mu^{d,d})^2+2\sum_{i=1}^{d-1}E_\mu^{d,i}E_\mu^{i,d}+\sum_{i=1}^{d-1}(E_\mu^{i,i}-E_\mu^{d,d}),
    \end{align}
    where the second line follows from the commutation relation $[E_\mu^{i,d},E_\mu^{d,i}]=E_\mu^{i,i}-E_\mu^{d,d}$.
    On the subspace $\mathrm V_{\mu-e_s}^{d-1}$, Eq.~\eqref{eq:ekk} gives
    \begin{align}
        E_\mu^{d,d}
        &=I_{\mathrm V_{\mu-e_s}^{d-1}},\\
        \sum_{i=1}^{d-1}E_\mu^{i,i}
        &=(|\mu|-1)I_{\mathrm V_{\mu-e_s}^{d-1}}.
    \end{align}
    Hence, we obtain
    \begin{align}
        \left. \left(\sum_{i,j=1}^{d}E_\mu^{i,j}E_\mu^{j,i}-\sum_{i,j=1}^{d-1}E_\mu^{i,j}E_\mu^{j,i} \right) \right|_{\mathrm{V}^{d-1}_{\mu - e_s}}
            &= \left(\sum_{k=1}^{d}\mu_k(\mu_k+d-2k+1) - \sum_{k=1}^{d-1}(\mu_k-\delta_{k,s})(\mu_k-\delta_{k,s}+d-2k)\right) I_{\mathrm{V}_{\mu - e_s}^{d-1}} \\
            &= \left(2c_s+|\mu|-d+1\right) I_{\mathrm{V}^{d-1}_{\mu-e_s}}\\
            &= \left(|\mu|+2\mu_s+d-2s-1 \right)  I_{\mathrm{V}^{d-1}_{\mu-e_s}}. 
    \end{align}
    Therefore, 
    \begin{align}
        c_s=\mu_s+d-s-1
        \leq \mu_1+d-2,
    \end{align}
    which completes the proof.
    
\bibliography{reference}